\documentclass[11pt]{article}
\usepackage{amsmath,amsthm,amssymb,amsfonts,mathtools}
\usepackage{algorithm,algorithmic}
\usepackage{graphicx}
\usepackage{xspace}
\usepackage{hyperref}
\usepackage[margin=1in]{geometry}
\usepackage{xcolor}

\newcommand{\Ms}{{\cal M}}
\newcommand{\Rs}{{\cal R}}
\newcommand{\Cs}{\boldsymbol{\mathcal C}}
\newcommand{\Es}{\boldsymbol{\mathcal E}}

\newcommand{\Ss}{{\bf S}}
\newcommand{\vs}{{\bf v}}

\newcommand{\Sos}{\boldsymbol{\mathcal{SO}}}
\newcommand{\supl}[2]{{\text supl}(#1\!\rightharpoonup\!#2)}      
\newcommand{\minprod}{\operatorname{minprod}}              
\newcommand{\mincons}{\operatorname{mincons}}              
 \newcommand{\fcomp}[1]{\xleftrightharpoons[]{#1}}
 
\newcommand{\prods}{\mathrm{prod}}
\newcommand{\supp}{\mathrm{supp}}
\newcommand{\clos}{\mathrm{clos}}
\newcommand{\js}{{\vee}}
\newcommand{\ms}{{\wedge}}
\newcommand{\Csr}{{\vec{\Cs}}}
\newcommand{\compsa}{\xleftrightharpoons[]{1}}

\newcommand{\GsE}{\boldsymbol{\mathcal G}(\Es)}
\newcommand{\Gs}{\mathbf{G}}          
\newcommand{\syns}{\rightsquigarrow}

\newcommand{\reqs}{\mathrm{req}}
\newcommand{\ERCs}{\mathrm{ERC}}

\theoremstyle{plain}
\newtheorem{thm}{Theorem}
\newtheorem{lem}{Lemma}
\newtheorem{defi}{Definition}

\newtheorem{cor}{Corollary}

\theoremstyle{remark}

\title{Synergy and Complementarity: The Generative Basis of Chemical Organizations}

\author{Tomas Veloz\textsuperscript{1} \and Alejandro Bassi\textsuperscript{2} \and 
$^1$Departamento de Matemáticas, Universidad Tecnológica Metropolitana, Santiago 833038, Chile\and
$^2$Faculty of Medicine, Universidad de Chile, Santiago, Chile
\texttt{tomas.velozg@utem.cl}
}

\date{}

\begin{document}
\maketitle

\begin{abstract}
Identifying structural features ensuring the persistence of a reaction network is fundamental for understanding the evolution and complexification of biological entities, from origins of life to metabolic engineering. By defining organizations as subsets of species that are closed and self-maintaining, Chemical Organization Theory (COT) shows that organizations enable filtering out the regions of the phase space where attractors can exist. This result and several variants of it open a promising path to study the constitutive, developmental, and adaptive aspects of metabolic and other kinds of complex reaction networks where traditional dynamical systems methods are intractable. Despite various theoretical explorations, current methods to compute the organizations are in general limited to dozens of species/reactions due to unclear combinatorial challenges.

In this article we study the problem of generating organizations as non-redundant combinations of smaller components. We first analyze the combinatorial challenges leading to redundant combinations. Second, we identify the minimal building components for non-redundant combinations, called Elementary Reaction Closures (ERC), and characterize them as a hierarchy. Third, we show that all non-redundant generative processes of organizations can be made aggregating ERCs that hold either a synergetic or complementarity relationship. These two notions are novel and open very interesting routes to understand the internal composition of complex metabolisms.  

We next show that synergies and complementarities scale far more slowly than the structures used in previous combinatorial approaches. By computing these relationships for $438$ biological reaction networks from the BioModels and BiGG databases, we obtain that ERCs grow sublinearly with reaction count while the proportion of synergies and complementarities shrinks as networks grow, confirming that larger biological networks achieve persistence through a progressively smaller set of irreducible generative relationships.

We do not prescribe algorithms to compute organizations, but discuss how algorithms based on our results could make COT applicable to reaction networks containing thousands of reactions using regular desktop computers.
\end{abstract}

\textbf{Keywords:} Reaction networks, Chemical organization theory, Persistence, Synergy, Complementarity, Novelty 

\section{Introduction}

The mathematical study of reaction networks has become increasingly important for understanding complex systems ranging from biochemistry to ecology and social systems~\cite{seminal,peter2021linking, veloz2021toward, veloz2022modelling}. A fundamental challenge in this field is identifying subnetworks capable of sustained operation within large reaction networks. This problem has direct applications in drug discovery, metabolic engineering, and systems biology because persistent sub-networks represent the parts of the reaction network that can operate independently in some dynamically stable conditions such as steady growth, self-maintenance, or periodic behavior. This operational conceptualization of persistence underlies theories of autocatalysis and emergence of life~\cite{peng2022hierarchical,hordijk2018autocatalytic,peter2023computing}, and in a more fundamental way represent the forms in which a system is assumed to be stable enough to be observed as such~\cite{chen2023applications,turanli2018network}. 

Chemical Organization Theory (COT) was introduced by Peter Dittrich and Pietro Speroni di Fenizio in~\cite{seminal} for identifying persistence from reaction network properties referring to its self-productive capacity.  The link between persistence and COT is that the set of species present in an attractor of the dynamics ---a fixed point and, under additional conditions, a periodic orbit or other persistent behaviour--- is a chemical organization~\cite{seminal,peter2011relation,peter2021linking}. This correspondence is not one-to-one: several attractors may share the same set of species, and an organization need not host an attractor for a given kinetics. Organizations therefore identify which sets of species can coexist in a persistent state, not how many such states exist. Since a chemical organization is a set of species, it defines a subspace of the phase space, and the organizations of a network provide a cover of all its attractors, whatever the kinetic law, parameters and initial conditions~\cite{peter2023computing}.  In a nutshell, a chemical organization is a set of species $X$ that is closed and self-maintaining: Closed means that no species in the products of the reactions that can be triggered by $X$ (referred to as $\prods(\Rs_X)$) is outside $X$, and self-maintaining means that there exists a positive-definite vector $\vs$ specifying the reaction relative rates of the reactions in $\Rs_X$ (also known as flux vector restricted to the set $X$) ensuring that the total production of the species in $X$ is non-negative. The latter is mathematically expressed by $\Ss\vs\geq \vec{0}$, where $\Ss$ is the stoichiometric matrix restricted to $X$.\\ 

Different to autocatalytic approaches, COT does not invoke catalysts or other special species for achieving self-production~\cite{hordijk2018autocatalytic,peter2021linking,peter2023computing}. Therefore, COT has been applied beyond biochemical systems~\cite{kreyssig2012cycles,peter2020structure}, including areas such as evolutionary dynamics of ecological systems~\cite{veloz2020complexity,veloz2021toward}, chemical computation~\cite{dittrich2007organization}, and others applications~\cite{veloz2014reaction,Veloz2017a,heylighen2024chemical,dittrich2008chemical,veloz2021analytic,veloz2022modelling}. 

From a computational perspective, verifying whether a particular sub-network is closed requires a lookup to the list of produced species by each reaction (linear time w.r.t. $|\Rs_X|$), while verifying self-maintenance is equivalent to a linear programming (LP) problem~\cite{dantzig1963linear} and hence it is superlinear w.r.t $|\Rs_X|$, but tractable for realistic reaction networks~\cite{chandru1998linear}. However, knowing whether a reaction network contains an organization at all is a NP-hard problem~\cite{centler2008computing} implying that leveraging from the theoretical power of COT poses a non-trivial challenge for biological reaction networks. The COT seminal paper~\cite{seminal} stated\\

{\centering \it  Identifying the computational complexity of algorithms able to find all organisations would be highly desirable. Especially it would be interesting to know, which kind of networks can be analysed with contemporary computers. Developing efficient algorithms for computing or approximating the set of organisations and estimating the computational complexity of these algorithms is a separate research topic for the future. So far we have implemented some preliminary algorithms. One of them computes the set of
organisations from the bottom up by starting from the smallest organisation and then, recursively, adding molecules in order to generated all the organisations above... }\\

Various approaches to compute organizations, mostly operating in a bottom-up fashion, have been proposed~\cite{centler2008computing}. These typically begin identifying closed sets, which is the easiest of the two conditions, and then filter out the self-maintaining ones. Unfortunately, bottom-up approaches typically explore a vast number of {\it irrelevant} closed sets~\cite{speroni2015lattice}. For example, combining two closed sets with no species in common leads to a closed set~\cite{veloz2019existence}. Typically, reaction networks admit an exponential number of such combinations. For example, early E. coli metabolic models containing around 600 reactions were still solvable using parallel computation~\cite{centler2010parallel}. However, modern genome-scale reconstructions include over 2,400 reactions\cite{orth2011comprehensive}, and human metabolism involves 8,000+ reactions, leaving modern biologically relevant networks intractable. Therefore, the emergence of omics integration and the growth of large-scale RN models in different areas such as medicine~\cite{selevsek2020network,turanli2018network,chen2023applications}, origins of life~\cite{preiner2020future,cuevas2023modular}, and interdisciplinary areas such as ecology~\cite{scott2024metabolic,mcdaniel2021prospects}, social sciences~\cite{veloz2014reaction,veloz2022modelling} and others~\cite{garrido2022integrating}, require a systematic analysis to explain when organizations are relevant from a generative perspective, so efficient algorithms can be implemented~\cite{peter2023computing}.

In this work we focus on the above-mentioned systematic analysis without attempting to build any efficient algorithm~\cite{centler2008computing,centler2010parallel,peter2023computing}. We develop a framework to understand the inner generative structure of reaction networks. In this sense, our work is better understood as a contribution to the logico-algebraic structures underlying biological organization~\cite{fontana1994arrival} rather than as a computational or dynamical analysis contribution. Specifically, we describe and characterize how self-production can be efficiently obtained by combining smaller non self-productive modules.

An important methodological step in this work, is that we frame our analysis on a concept of self-production that is simpler than self-maintenance: A set is {\it semi-self-maintaining} (SSM) when its set of reactions $\Rs_X$ produces all species that are consumed. Closed and SSM sets are called {\it semi-organizations} (SO). It is known that every organization is a SO, but the converse is not always true~\cite{seminal}. Note that verifying SSM is as simple as verifying closure (a lookup to the produced and consumed species by each reaction). By this methodological choice, we step aside from the LP-related challenges and opportunities for identifying organizations~\cite{Veloz2017b}. Once SOs are computed, we can verify which of these are organizations using linear programming and refine such candidates to stability using other methods~\cite{peter2010feasibility}.

Note that for a single set verifying if it is a SO is trivial, but we cannot run such verification set by set in a large reaction networks. Therefore, our core problem is understanding how the productive structure (what is consumed and produced) of a combination of closed sets can be deduced from the productive structure of the combined closed sets. When two or more sets are combined, its closure is not necessarily the union of the sets, novel reactions can occur that take reactants from the different sets in the combination (we call this {\it synergy}), and it can also be the case that what is produced by one set is needed by the other set to reach SSM (we call this {\it complementarity}). These two concepts define what we call {\it productive novelty} for a combination of sets, and are crucial to our analysis. The core idea of this paper is that productively novel generation defines a criteria of {\it relevance} to SOs that ensures minimal combinatorial explosion. Any other SO is {\it irrelevant} because its productive structure will be the direct union of the productive structure of its parts, and hence it is not relevant to generate it.

More technically, we show that the closed sets generated by the reactants of each single reaction --called Elementary Reaction Closures (ERC)-- correspond to the minimal building blocks necessary to generate all relevant SOs. Moreover, productive novelty for a combination of closed sets is equivalent to saying their closure and SSM cannot be simultaneously known directly from the properties of the sets combined, i.e. some computation is needed to identify their closure and/or semi-self-maintenance. We operationalize productive novelty by formalizing synergy and complementarity: Two modules are synergetic if when combined they trigger a novel reaction that cannot be triggered by the parts, and complementary when there exists a species produced by one set that is consumed but not produced, i.e. {\it required}, by the other (this is properly formalized in sections~\ref{Synergy} and~\ref{Complementarity}). 

Our formalization demonstrates that a special class of synergies and complementarities --so called fundamental-- are sufficient to build all relevant SOs. The importance of this result does not translate directly into an algorithm, but by the fact that the structure from which we can generate all relevant SOs is tiny compared to previously known combinatorial structures to compute organizations. By applying brute-force enumeration of synergies and complementarities across $438$ biological reaction networks from the BioModels and BiGG databases we estimate very favourable scaling laws. The ERC count itself already compresses the reaction network sublinearly (exponent ${\approx}0.71$), and the proportion of \emph{fundamental} synergies and complementarities among all ERC pairs shrinks further as networks grow, even as raw counts increase --- confirming that larger networks achieve persistence through an increasingly selective set of irreducible generative relationships. These results strongly motivate algorithms that exploit this structure, which could make COT applicable to networks with thousands of reactions.

It is important to mention that we introduce a new and conceptually dense formal framework that has no counterpart in reaction network theory. In glossary on table~\ref{glossary} we present the terminology, with the exception of some elements that are introduced and discussed more thoroughly. 

\begin{table}[h!]
\centering
\small
\begin{tabular}{|p{3.1cm}|p{2.8cm}|p{9.5cm}|}
\hline
\textbf{Term} & \textbf{Symbol} & \textbf{Meaning} \\
\hline
Reaction Network & $(\Ms,\Rs)$ & Species  $\Ms=\{s_1,...,s_n\}$, reactions $\Rs=\{r_1,...,r_k\}$. \\
Number of species & n & $|\Ms|=n$\\
Number of reactions & k & $|\Rs|=k$\\
Support & $\supp(r)$ & Necessary species to trigger a reaction $r\in \Rs$.\\
Activable  Reactions & $\Rs_X$ & Reactions that can be triggered from $X$.\\
Produced& $\prods(\Rs_X)$ & Species produced by $X$.\\
Required  &$\reqs(\Rs_X)$& $\supp(\Rs_X)\setminus\prods(\Rs_X)$: Species consumed but not produced by $X$.\\
Closure  & $\clos(X)$ & End result of $X_n=X_{n-1}\cup\prods(\Rs_{X_{n-1}})$, with $X_0=X$.\\
Join& $X\js Y$& $\clos(X\cup Y)$ the joint closure.\\
Reactive  & $\Csr$ & Closed sets whose species participate in at least one reaction. \\
Connected & $\Csr_{con}$ & A reactive closed set whose species are all connected via reaction pathways. \\
ERC & $\Es$ & Elementary Reaction Closures: $\clos(\supp(r))$ for a single $r$.  \\
SSM & $\reqs(X)=\emptyset$ & Semi-self-maintaining: Every consumed species in $\Rs_X$ can be produced by $\Rs_X$.\\
Connected Semi-Organizations & $\Sos$& Connected, closed and semi-self-maintaining sets of species.\\
Direct containment & $E_2\prec E_1$ & $E_2\subset E_1$ with no ERC strictly in between; edges of the ERC hierarchy (Def.~\ref{def:direct_containment}).\\
Generator & $\Gs=\{E_1,\dots,E_l\}$ & $\Gs$ is a set of ERCs.\\
Set of generators & $\GsE$ &Power set of $\Es$ modulo closure (Def.~\ref{ERC_generators}).\\
Generated set & $\bigvee\Gs$ &  Join of the elements in the generator denoted by $\bigvee\Gs=X$\\
Synergy & $\{E,E'\}\syns E_{syn}$ & Combination of ERCs that activates another ERC (fundamental synergy shown). \\
Degree of a synergy & $d$ & Number of ERCs jointly needed to trigger the synergetic reaction (Def.~\ref{def:degree}).\\
Supply &$\supl{E}{E'}$ & $\prods(\Rs_E)\cap\reqs(E')$: species produced by $E$ and required by $E'$ (Def.~\ref{def:supply}).\\
Fundamental complementarity & $E\fcomp{s}E'$ & A relationship where a minimal ERC produces a species $s$ that another minimal ERC requires (Def.~\ref{def:fundamental_complementarity}).\\
Elementary SO & $\Sos_0$ & Elementary semi-organization: a semi-organization with no semi-organization within (generative index $0$).\\
Higher-order SO & $\Sos_i,\ i>0$ & Semi-organization of generative index $i$: it contains semi-organizations of lower index. \\
\hline
\end{tabular}
\caption{Glossary of key terms and symbols used throughout the paper.}
\label{glossary}
\end{table}

The paper is organized as follows: In section~\ref{RN} we introduce reaction networks, the basics of COT and analyze the current approaches and understanding of the problem of building semi-organizations. In sections~\ref{ERC}, we introduce Elementary Reaction Closures and its hierarchical structure. Next, in sections~\ref{Synergy}, and~\ref{Complementarity} we introduce the Synergy and Complementarity relations, and show that they can be studied directly from the Elementary Reaction Closures (ERCs) without loss of generality. The latter is formalized through a number of results that clarify how the inner generative structure of semi-organizations operates within the larger reaction network. We conclude our theoretical construction in section~\ref{ProductiveNovelty} by showing that the ERC hierarchy, equipped with highly restrictive forms of synergy and complementarity, so called fundamental, are sufficient to build all semi-organizations of a reaction network with minimal sequences of ERC combinations. Finally, we analyze the statistical and scaling structure of our new concepts using biological networks in the BioModels~\cite{BioModels2020} and BiGG~\cite{schellenberger2010bigg} databases, and discuss a series of challenges and opportunities that these results open up for the future of reaction network modeling and analysis.  

\section{Reaction Networks and Chemical Organizations}
\label{RN}
\subsection{Preliminaries}
A reaction network (RN) consist of a set of species, representing the primary entities our system is made of. In chemistry these are assumed to be chemical elements or molecules. For other applications these could be subatomic particles, information resources, ecological species, beliefs,  etc.~\cite{seminal,Veloz2017a,heylighen2023modeling}. A {\it reaction} is triggered when specific collections of species, so-called reactants, encounter. The result is that the reactants become another collection of species, so-called products. 

A RN can be described as a set of species $\Ms=\{s_1,\ldots,s_n\}$ and a set of reactions $\Rs=\{r_1,\ldots,r_k\}$, where each reaction  $r_i$ is denoted as:

\begin{equation}
r_i=a_{i1}s_1+\cdots + a_{in}s_n\to b_{i1}s_1+\cdots + b_{in}s_n,
\label{eq:reaction}
\end{equation}
with $a_{ij} \in \mathbb{R}_{\geq 0}$ and $b_{ij}\in \mathbb{R}_{\geq 0}$ for $j=1,...,n$. Hence, a species $s_j \in \Ms$ is called a reactant or product of $r_i$ if $a_{ij}>0$ or $b_{ij}>0$ respectively. For each $r\in \Rs$ we define the support of $r$ as the set of reactants of $r$ and denote it by $\supp(r)$ . Similarly, $\prods(r)$ is the set of products of $r$. When a RN is initialized from a set of species $X$, we denote by $\Rs_X \subseteq \Rs$ the set of reactions whose support is in $X$, and extend naturally the notions of support and product to a set of reactions by $$\supp(\Rs_X)=\bigcup\limits_{r\in \Rs_X} \supp(r),\quad \prods(\Rs_X)=\bigcup\limits_{r\in \Rs_X} \prods(r).$$

Note that when $a_{ij}=0$ for all $j=1,...,n$ the reaction has no reactants, and hence it is called inflow. Inflow reactions are the equivalent to what in other reaction network frameworks are called food sets~\cite{hordijk2018autocatalytic}. COT introduces the exchange of matter with the environment through inflow and outflow reactions, i.e. those with $b_{ij}=0$ for all $j=1,...,n$.

\begin{defi}
$X$ is \textit{closed} if and only if $\prods(\Rs_X) \subseteq X$.
\label{def:closed}
\end{defi}

Closed sets of species cannot generate novel components from their reactions.  Note that species generated by inflow reactions are part of each and every closed set. This is because their set of reactants is the emptyset, and hence triggerable in every set we consider to build a closure.

\begin{defi}
A set $X$ is semi-self-maintaining (SSM) if and only if its required set of species $\reqs(X)=\supp(\Rs_X)-\prods(\Rs_X)$ is empty.
\label{SSM}
\end{defi}

SSM sets are able to produce all species that they consume, but not necessarily at the same quantity or rate than they are consumed. 

\begin{defi}
A set of species $X\subseteq \Ms$ that is closed and semi-self-maintaining is called a semi-organization (SO). 
\end{defi}

SOs represent a weak but easy to compute form of persistence. Both closure and SSM are properties that can be verified in linear time w.r.t. $|\Rs_X|$. In order to clarify the value of this choice,  we complete our introduction to COT.

We define the stoichiometric matrix $\Ss$ by $\Ss_{ij}=b_{ij}-a_{ij}$, and determines the net production of each species $s_j$ by each reaction $r_i$, $j=1,...,n$; $i=1,...,k$. 

\begin{defi}
Let $\Ss$ be the stoichiometric matrix of $(\Ms,\Rs)$. $X\subseteq \Ms$ is self-maintaining if and only if there exists a vector $\vs$ such that for all $i=1,...,k$ the $i$th coordinate $\vs[i]$ satisfies
\begin{equation*}
\begin{cases}
\vs[i]>0 & \text{if } r_i \in \Rs_X, \\
\vs[i]=0 & \text{if } r_i \notin \Rs_X.
\end{cases}
\end{equation*}
and 
$$\Ss\vs\geq \vec{0}$$
\label{self-maintainance}
\end{defi}

Self-maintaining sets can implement reaction pathways, represented by a vector $\vs$, that on the one hand trigger all reactions in $\Rs_X$ in a strictly positive rate, and the result of such process regenerates all used components. Computationally speaking, deciding self-maintenance of a set of species is a Linear Programming problem~\cite{dantzig1963linear}. 

\begin{defi}
A set of species $X\subseteq \Ms$ that is closed and self-maintaining is called an organization. 
\end{defi}
It has been shown that the set of species with positive concentration at any fixed point of a dynamical system built upon a reaction network is an organization~\cite{seminal}; the converse does not hold, and several fixed points may share the same organization. Further extensions that make precise the criteria and conditions to relate chemical organizations to other persistent dynamics such as periodic orbits and limit cycles in both ordinary or partial dynamical systems have been developed in~\cite{peter2010feasibility,peter2011relation,Veloz2017b,peter2021linking}. We are not going to provide in all detail such correspondence here, but clarify that COT characterizes the subspaces where persistence dynamics can exist~\cite{peter2023computing}. However, it does not specify what parametric conditions ensure persistent dynamics. Therefore, COT provides a qualitative picture of the stable dynamics of a system. This is particularly useful when dealing with large reaction networks, where precise dynamical methods cannot be applied, and questions regarding structural evolution are equally or more important than the specific stable configurations.

COT has been successfully applied to multiple areas of research related to biochemistry as well as interdisciplinary applications~\cite{velegol2018chemical,kovacs2020sustainability,veloz2021analytic} (see introduction for other references). Organizations differ from approaches such as elementary modes~\cite{schuster1994elementary}, T-invariants~\cite{murata2002petri}, and other models of stability in that all reactions in $\Rs_X$ are required to occur, and, unlike autocatalytic sets, they do not require catalysts~\cite{hordijk2018autocatalytic}. In this sense, chemical organizations represent an abstract model of persistence based on self-production under dynamical conditions where "all that can possibly occur must occur in a positive rate"~\cite{Veloz2017a}.  

It is trivial that every organization is a semi-organization, and although no guarantee is given that the set of semi-organizations is a good approximation of the set of organizations~\cite{centler2008computing}, in most cases the number of semi-organizations is typically much smaller than the number of closed sets. 


\subsection{Generating semi-organizations}
\label{connectedness_rel}

Algorithms to compute semi-organizations of a RN have approached different strategies for generating semi-organizations. The natural first step is to build closed sets.

\begin{defi}
Let $X \subseteq \Ms$, we define the \textit{generated closure} $\clos(X)$ of $X$, as the smallest closed set containing $X$. 
\label{def:genclos}
\end{defi}
Note the closure $\clos(X)$ of a set of species $X$ is unique~\cite{seminal}. This becomes evident in the procedure specified in table~\ref{glossary} that begins with $X=X_0$ and adds $\prods(X)$ to generate a new set $X_1=X_0\cup\prods(X_0)$ and continues the addition of products of $X_1$, and of subsequent extensions recursively until closure is reached. This process is deterministic as it is defined by the reaction network's structure, and ensured to end because $\Ms$ is finite, but it typically ends before $\Ms$ is reached. In fact, reaction networks typically have several closed sets of different sizes being $\Ms$ the largest.

\begin{defi}
The collection of all closed sets $\mathcal{C}$ of a RN $(\Ms,\Rs)$ is called the \textit{closed sets structure}.
\label{def:enclosed}
\end{defi}

A set $X$ of species that does not trigger any reaction is trivially closed, and it is also trivially SSM. Both because $\reqs(\Rs_X)=\prods(\Rs_X)=\emptyset$.

More generally, given a semi-organization $X\subset \Ms$, for every species $s$ that does not react with any species in $X$, i.e. $\Rs_{X\cup \{s\}}=\Rs_{X}$, we obtain a new semi-organization because neither closure or self-maintenance is affected~\cite{veloz2019existence}. 

\begin{defi}
A species $s\in \Ms$ is reactive w.r.t to $X$ if and only if there exists $r \in \Rs_{X\cup s}$ such that $s\in \supp(r) \cup \prods(r)$. A set of species $X$ is reactive if and only if each species $s\in X$ is reactive w.r.t to $X$. We call $\Csr$ the set of reactive closed sets.
\label{def:reac_set}
\end{defi}

From a generative point of view, building non-reactive closed sets is irrelevant. In fact, its dynamical properties are exactly the same with or without the non-reactive species. The efficiency gain in not computing non-reactive closed sets is huge: the number of non-reactive closed sets that one can build from a closed set $X$ is exponential with respect to total of non-reactive species with respect to $X$.

Consider for example a reaction network with three kinds of species: species $\{s_i\}$ that reproduces by consuming a food source $\{f_i\}$ specific to itself, and produces a residual $y_i$ that is also specific, for $i=1,...,n$. Hence $\Ms=\{f_i,s_i,y_i\}$ with $i=1,...,n$, and reactions 
\begin{equation}
r_{i}=f_i+s_i\to 2s_i+y_i\text{, for } i=1,...,n.
\label{RN-Ex0}
\end{equation}

Let $X_i=\{f_i,s_i,y_i\}$ for $i=1,...,n$. Note that $X_i\in\Csr$ for $i=1,...,n$, and since each of these closed sets is non-overlapping with all others, we can build a new closed set by adding to $X_i$ any closed subset of the remaining species. For each $j\neq i$, the closed subsets of $\{f_j,s_j,y_j\}$ are all its $2^3=8$ subsets except $\{f_j,s_j\}$, which triggers $r_j$ but does not contain $y_j$. Therefore, $7^{n-1}$ different closed sets contain each $X_i$. These closed sets typically contain several non-reactive species. Therefore, $\Csr$ is expected to be exponentially smaller than $\Cs$. 

Despite the significant reduction brought by focusing on reactive closed sets, $\Csr$ still contains several elements that are irrelevant from a generative perspective. Consider for example $X=X_i\cup X_j$ with $i\neq j$. In this case, $X$ is closed, but the reason is that both $X_i$ and $X_j$ are closed and they do not couple in any way within $X$. Therefore, $X$ does not exhibit any novelty with respect to its parts~\cite{veloz2019existence}. 
\begin{defi}
Two species $s_1,s_2\in X$ are directly connected w.r.t to $X$ if and only if there exists $r \in \Rs_X$ such that $\{s_1, s_2\} \subseteq \supp(r) \cup \prods(r)$.\\
Two species $s_a,s_b$ are connected if there is a sequence of species $s_1,...,s_l$ where $s_1=s_a$, $s_l=s_b$ and $s_i$ is directed connected to $s_{i+1}$ for all $i=1,...,l-1$.\\
A set of species $X \in \Csr$ is connected if and only if each pair of species $s_1,s_2\in X$ is connected w.r.t to $X$. Moreover, we call $\Csr_{con}$ the set of connected closed sets.
\label{def:conn_set}
\end{defi}

Note that in our the example in Eq.~\eqref{RN-Ex0} every combination of $\{X_i\}$ is reactive but not connected, implying $|\Csr_{con}|=n$, $|\Csr|=2^n$, and $|\Cs|=7^n$.\\

First efforts to identify semi-organizations of a RN computed connected closed sets bottom up while checking self-maintenance   ~\cite{centler2008computing,centler2010parallel}. Unfortunately, the combinatorial explosion of generative paths leading to connected semi-organizations makes these approaches intractable for moderately large networks, even using parallel computation~\cite{centler2010parallel}. Alternatively, initial attempts tried identifying the smallest self-producing reaction pathways, better known as elementary modes or T-invariants~\cite{schuster1994elementary,grafahrend2008modularization}, and combine these pathways to discover closed sets. This approach proved less efficient than the former approach in several cases, especially when the species can be produced by several linearly independent pathways~\cite{centler2008computing}.  More recent attempts to compute semi-organizations have tried to harness the mathematical structure of closed sets and exploit their properties to simplify their construction.

\subsection{Order-Theoretical Approach to Generate semi-organizations}
\label{LatticeClosedSets}
Let $X,Y\subseteq\Ms$ and 
\begin{equation}
X \js Y=\clos(X\cup Y),\quad \text{and} \quad X \ms Y=\clos(X\cap Y),
\label{join-meet}    
\end{equation}
be called {\it join} and {\it meet} of $X$ and $Y$ respectively. In~\cite{seminal} it is shown that the triplet $(\Cs, \js, \ms)$ is a {\it lattice}, meaning a partial ordered set that is associative and commutative with respect to $\js$ and $\ms$, and the following identities hold:
\begin{equation}
\begin{split}
X\js (X\ms Y)&=X,\quad X\ms (X\js Y)=X\quad (\text{absorption laws}),\\
X\js X&=X,\quad X\ms X=X\quad (\text{idempotent laws})
\end{split}
\label{Lattice-axioms}   
\end{equation}
Note that if we identify two closed sets, we can combine them using $\ms$ and $\js$ to build two other closed sets. Therefore, these operators can be exploited to produce generative algorithms to compute $\Cs$~\cite{speroni2015lattice,peter2023computing}. Although $\ms$ is necessary to equip $\Cs$ with the lattice structure, from a generative perspective it is not a useful operator, as $\js$ is sufficient to build $\Cs$ (we will show this in detail). Therefore, the actual structure that we will be considering is a semi-lattice~\cite{gratzer2002general}. In order to show the inefficiencies induced by the use of $\ms$ consider the reaction network $$r_1=a+b\to c;\quad r_2=a+d\to e.$$
Note that $X_1=\{a,b,c\}$ and $X_2=\{a,d,e\}$ are in $\Csr_{con}$ but $X_1\ms X_2=\{a\}\notin \Csr_{con}$. Indeed, $X_1\ms X_2$ is trivially closed because it is non-reactive.   

In summary, the difficulties in identifying SOs stem from its combinatorial structure: It is not clear how closed and SSM subnetworks emerge as combinations of smaller modules. Previous algorithms to compute organizations first compute semi-orgs and then verify self-maintenance~\cite{centler2008computing,peter2023computing}, but it is in this first filtering where the computation of organizations becomes unfeasible.

\begin{defi}
Let ${\Sos}$ be the set of connected semi-organizations. 
\end{defi}

\begin{lem}
$\Sos \subseteq \Csr_{con}\subseteq \Csr\subseteq \Cs$
\label{Sos_subset}
\end{lem}

The goal of this paper is to develop an efficient generative structure for computing $\Sos$. 

\subsection{Efficient Generation and Productive Novelty}
\label{EfficientGeneration}

\begin{defi}
Let $X\in\Csr$. We say that a set of closed sets $\Gs=\{X_1,...,X_{l}\}$ with $X_i\in \Csr$ for $i=1,...,l$ generates $X$ if and only if 
$$\bigvee\Gs\;=\;\bigvee\limits_{X_i\in \Gs} X_i \;=\; X.$$
We also say $\Gs$ is a generator of $X$, and call $\bigvee\Gs$ the closed set reached by $\Gs$.
\label{cover_def}
\end{defi}

A given set of closed sets generates a single closed set, but different sets of closed sets can generate the same closed set. For example, every closed set $X$ is generated by the singleton $\{X\}$, by the set of all closed sets it contains, and by a combinatorial number of intermediate collections. Since $\js$ is commutative, associative and idempotent, a generator is fully determined by the collection of closed sets it comprises, so we treat generators as sets throughout. 


Adding a closed set $Z$ to a generator $\Gs$ is written as the set union $\Gs\cup\{Z\}$, and combining two generators $\Gs$ and $\Gs'$ is written $\Gs\cup \Gs'$.

\begin{lem}
Let $\Gs=\{X_1,...,X_L\}$ be a generator of $X=\bigvee\Gs$. $X\in\Sos$ if and only if $$\bigcup_{i=1}^L\supp(\Rs_{X_i})\subseteq\prods(\Rs_X)$$ 
\label{SO_condition_gen}
\end{lem}
\begin{proof}
$\Rightarrow:$ $X\in\Sos$ implies $\supp(\Rs_X)\subseteq\prods(\Rs_X)$. Since $\bigcup_{i=1}^L\supp(\Rs_{X_i})\subseteq \supp(\Rs_X)$ the statement follows.

$\Leftarrow:$ Since $\Gs$ is a generator, this means that $\bigvee_{i=1}^L X_i=X.$ This implies that every species $s\notin \bigcup_{i=1}^L\supp(\Rs_{X_i})$ must be in the products of some reaction in $\Rs_{X}$ (otherwise it would not be in $X$). Therefore, it is enough that $\bigcup_{i=1}^L\supp(\Rs_{X_i})\subseteq\prods(\Rs_X)$ to ensure $\supp(\Rs_{X})\subseteq\prods(\Rs_X)$, i.e. $X\in \Sos$.
\end{proof}

Lemma~\ref{SO_condition_gen} shows that verifying semi-self-maintenance for a closed set is reduced to verifying the semi-self-maintenance of its generator. Efficient generation of the whole set of semi-organizations asks us to push this same reasoning. We propose the following efficient generation desiderata:

\begin{enumerate}
\item Generators are built from a minimal collection of closed sets $\Es$: This means $X\in\Sos$ implies there exists a generator $\Gs$ with $\bigvee\Gs=X$ and elements in $\Es$ only. 
\item Generators are efficiently extended: This means, given $\Gs$ that does not generate a SO, we have a criteria to choose which elements from $\Es$ shall be used to extend $\Gs$ for it to generate a SO. This avoids brute force exploration of the power set of $\Es$ to build generators. 
\end{enumerate}
 
To clarify the desiderata, suppose we have a generator $\Gs$ with $\bigvee\Gs=X\notin \Sos$. By Lemma~\ref{SO_condition_gen} we know $X$ consumes a species it does not produce. Thus, we consider extending it with a candidate closed set $Y$ to form $\Gs'=\Gs\cup\{Y\}$ aiming to discover $\bigvee\Gs'\in\Sos$. Note that if $Y\subseteq X$ then $\bigvee\Gs'=X$, so we violate criteria 2. We assume otherwise. The ideal scenario is that $Y$ produces all the species that $X$ cannot produce, but this cannot be known until $\bigvee\Gs'$ is computed. However, we know that two criteria have the potential to advance in the desired direction: either the encounter of $X$ and $Y$ \emph{activates reactions that neither could trigger alone}, formally $\Rs_{X'}\supsetneq \Rs_{X}\cup\Rs_{Y}$, opening productive possibilities absent from both parts; or one side \emph{produces a species that the other requires}, formally $\prods(\Rs_{X})\cap\reqs(Y)\neq\emptyset$ (or $\prods(\Rs_{Y})\cap\reqs(X)\neq\emptyset$ ), so that the combined deficit shrinks. We call the first case \emph{synergy} and the second \emph{complementarity}, and we refer more generally to both of them as \emph{productive novelty}.

These properties are the formal counterparts of two pervasive features of living matter --- the emergence of metabolic function from the joint action of components that are inert in isolation, and the cooperative cross-feeding by which one subsystem supplies what another lacks. 

\section{Elementary Reaction Closures}
This section is devoted to identify the set $\Es$ in criteria 1., and its properties. An elementary reaction closure (ERC) is the closure of the species required to trigger a single reaction. 
\label{ERC}
\begin{defi}
Let $r\in \Rs$. We define $\ERCs(r)=\clos(\supp(r))$ as the elementary reaction closure (ERC) of $r$. Moreover, we denote the set of all ERCs by $\Es$.
\end{defi}
 ERCs are the smallest possible elements of $\Csr_{con}$. Indeed, given a reaction $r$, any set $Y\subset \supp(r)$ is either the support of another reaction $r'$ (as for example the case $r=a+b\to c$ and  and $r'=a\to d$), hence generating another ERC equals to $\clos(Y)$, or $Y$ does not trigger any reaction and hence it is non-reactive, i.e. $Y\notin\vec{\Csr}$. 

Note that different reactions can produce the same ERC. For example, the set of reactions forming a loop
$$r_0=s_0\to s_1,~r_1=s_1\to s_2,~\cdots\ ,~r_l=s_l\to s_0,$$ generates the same ERC $ERC(r_i)=\{s_0,...,s_l\}$ for $i=1,...,l$.\\
By introducing the relation $r\sim r'$ iff $\ERCs(r)=\ERCs(r')$ we obtain that each ERC is an equivalence class of reactions under $\sim$, and hence a partitioning of the reactions into non-overlapping classes generating the same closed set~\cite{peter2023computing}. This helps to simplify the representation of the reaction network in two senses. First, ERCs are productively different when looking from the closure point of view, and the number of ERCs is always smaller or equal than the number of reactions, simplifying the reaction network generative properties.

Elementary Reaction Closures (ERCs) were introduced under the name of {\it basic sets} by the authors of this paper in~\cite{veloz2019existence} and later conveniently renamed by Peter et al. to build an algorithm to generate distributed organizations~\cite{peter2023computing}. We remark that the methods we develop in this article can be applied to compute distributed organizations and other relevant structures such as (reflexive) autocatalytic sets~\cite{hordijk2018autocatalytic}. Such applications however, are out of the scope of this paper because in this article we concentrate on unveiling how to avoid the combinatorial intricacies that are found when computing semi-organizations only, a problem common to both computing organizations and distributed organizations~\cite{seminal, centler2008computing,speroni2015lattice,peter2023computing}. 

\subsection{Persistent ERCs: Inflow and co-production}

ERCs are certainly closed, but not persistent in general. Therefore characterizing when they are persistent deserves some clarification. 

\begin{defi}
We say $E\in\Es$ is persistent or a P-ERC if $E\in \Sos$.
\end{defi}

An important remark here is needed. In COT, a reaction specifying the inflow of a species $s$ to the system, or what is called a 'food species' in autocatalytic networks literature~\cite{hordijk2018autocatalytic}, or input places in Petri Nets literature~\cite{murata2002petri}, is represented by a reaction $\emptyset\to s$. Thus, we introduce a P-ERC $E_\emptyset$ which is associated to the closure of all inflow reactions.

\begin{defi}
We call $E_\emptyset=\clos(\emptyset)$ the closure of the inflow reactions. If a reaction network does not have inflow then $E_\emptyset=\emptyset$.
\end{defi}

This ERC is special because it is a P-ERC and also for every other ERC $E$ we have that $E_\emptyset\subset E$. Therefore, $E_\emptyset$ lies within each and every closed set (see fig.~\ref{ERC_hierarchies} right). Therefore, $ E_\emptyset$ is redundant to every generator and hence we will not mention it unless necessary. 

The following result characterizes the inner structure of P-ERCs.

\begin{lem}
$E\in\Es$ is a P-ERC if and only if for all $r\in \Rs_E$ we have that $s\in\supp(r)$ implies either $s\in E_\emptyset$ or there exists $r'\in \Rs_{E}$ for which $s\in\prods(r')$
\end{lem}
\begin{proof}
 $\Rightarrow$: Since $E$ is a P-ERC we have  $\reqs(E)=\emptyset$. We have two cases. The first case is $\supp(\Rs_E)=\emptyset$: every $r\in\Rs_E$ is an inflow reaction, hence there is no $s\in\supp(r)$ and the implication holds trivially. Indeed, in this case $E=E_\emptyset$ and every $s\in E$ belongs to $\prods(r)$ for some $r\in\Rs_E$. The second case is that $\supp(r)\neq\emptyset$ for some $r\in\Rs_E$; then $s\in\supp(r)$ implies $s\in\prods(r')$ for some $r'\in\Rs_E$, since otherwise $s\in\reqs(E)$, contradicting that $E$ is a P-ERC.\\
 $\Leftarrow$: Since for every reaction we have that its support is in the products of a reaction (being or not inflow) the condition of SSM is fulfilled directly.
 \end{proof}

\subsection{The ERCs Hierarchy}

ERCs have been already proposed as a promising approach to produce generative methods to compute all semi-organizations of a reaction network~\cite{peter2023computing}. However, a crucial and unexplored feature of ERCs is that some ERCs contain others, and hence we can think in some cases {\it some ERC generate other ERCs}. More formally, we say $\Es$ forms a partial order. As an example, consider the RN $(\Ms=\{s_1,...,s_4\},\Rs=\{r_1,...,r_5\})$, with 
\begin{equation}
\begin{split}
r_1&=s_1\to s_2, r_2=s_2\to s_1,\\
r_3&=s_3\to 2s_3,\\ 
r_4&=s_1+s_3\to s_4, r_5=2s_4\to s_3+s_1.
\label{RN-Ex1}
\end{split}
\end{equation}
Note that 
\begin{itemize}
    \item $E_1=\ERCs(r_1)=\ERCs(r_2)=\{s_1,s_2\}$,
    \item $E_2=\ERCs(r_3)=\{s_3\}$,
    \item $E_3=\ERCs(r_4)=\ERCs(r_5)=\{s_1,s_2,s_3,s_4\}$,
\end{itemize}
$E_3$ contains both $E_1$ and $E_2$, but $E_1\not \subseteq E_2$ and $E_2\not \subseteq E_1$. Therefore, $\Es$ is a partial order in general. 

\begin{defi}
Let $E_1$ and $E_2$ be two ERCs such that $E_2\subset E_1$. We say $E_2\prec E_1$ is a direct containment if and only if there is not $E\in \Es$ such that $E_2\subset E$ and $E\subset E_1,$ and the pair of all ERCs and its direct containments is denoted as $(\Es,\prec)$.
\label{def:direct_containment}
\end{defi}

In order to remark the non-triviality of $(\Es,\prec)$, we show the reaction network and ERC hierarchy for a BioModels reaction network with 26 species and 31 reactions (Biomodel 237, see figure~\ref{ERC_hierarchies} left and right respectively). Note that the number of ERCs is $21$, nearly a third smaller than the number of reactions. The latter reflects ERC is somehow a compression of the productive structure of the reaction network. However, the most remarkable feature is that there is significant depth and branching, and it cannot be noticed from the reaction network structure. Namely, the ERC hierarchy of this example contains two chains of length four, one chain of length three, two chains of length two, and two sub-hierarchies of depth and branching two. All chains end in $E_\emptyset$ as expected. Moreover, note that certain ERCs are closed (cyan), others are semi-organizations (orange), and others are organizations (green).  The latter implies that the ERC hierarchy contains already important information that can be used to generate semi-organizations. We will analyze such structure in detail in what follows.

\begin{figure}[h]
\centering
\includegraphics[width=\textwidth]{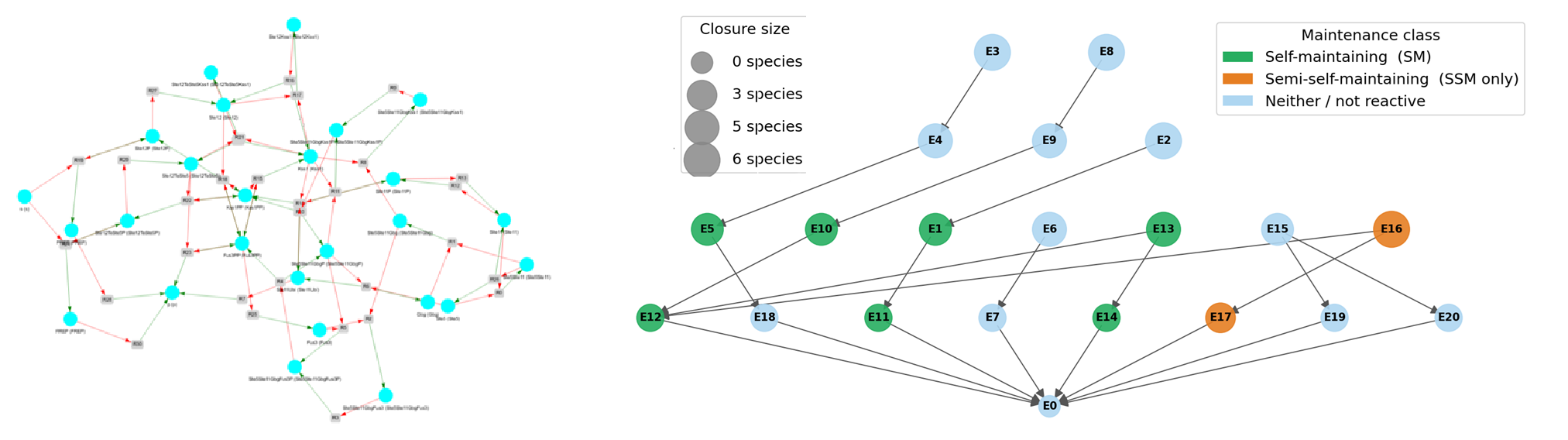}
\caption{ERC hierarchy example for Biomodel 237 {\it 'Schaber2006\_Pheromone\_Starvation\_Crosstalk'}. Left: Reaction network (species names invisible but irrelevant). Right: ERC Hierarchy.}
\label{ERC_hierarchies}
\end{figure}

\subsection{ERC-Generators, Irreducibility and Generative Relevance}
\label{inner_structure}
\begin{defi}
A generator $\Gs=\{E_1,...,E_l\}$ where $E_i\in\Es$ is called an ERC-generator of $X$.\\
The set of ERC-generators modulo closure, i.e. the power set $\mathcal{P}(\Es)$ in which two generators $\Gs,\Gs'$ are identified whenever $\bigvee\Gs=\bigvee\Gs'$, is denoted by $\GsE$.
\label{ERC_generators}
\end{defi}
We keep the notation $\GsE$ rather than $\mathcal{P}(\Es)$ to stress that many elements of $\mathcal{P}(\Es)$ generate the same closed set. This redundancy is precisely what motivates introducing synergy and complementarity in the next sections remove. These are the source of the combinatorial reductions, as we report in Section~\ref{RelevantExtensions}.

From a graph-theoretical perspective $(\Es,\prec)$ is a directed acyclic graph, namely the Hasse diagram of the partial order $(\Es,\subseteq)$, and from a lattice theoretical perspective $(\GsE,\subseteq,\js)$ is an atomic semi-lattice with $E_\emptyset$ as a bottom element~\cite{gratzer2002general}. We now show that $\GsE$ contains all elements of $\Csr_{con}$.
\begin{lem}
\label{lem:ERC_cover}
Every $X\in\Csr_{con}$ admits a generator in $\GsE$.
\end{lem}
\begin{proof}
Since $X$ is i) reactive and ii) closed, then i) all its species participate as reactants or products in at least one reaction in $\Rs_X$, and ii) no produced species by reactions in $\Rs_X$ are outside $X$. Therefore $X=\bigcup\limits_{r\in\Rs_X} \ERCs(r)$. Thus $\Gs=\{\ERCs(r): r\in \Rs_X\}$ is a ERC-generator of $X$.
\end{proof}
Lemma~\ref{lem:ERC_cover} shows that all relevant closed sets can be obtained by combining ERCs only. This is very important as it indicates that in order to generate semi-organizations we can only focus on ERC combinations. 

However, $\GsE$ also generates irrelevant closed sets, for example by combining two ERCs that have no species in common. This will lead to a disconnected closed set. Therefore, we now turn our attention to understand how ERCs shall be combined to efficiently generate $\Sos$. Coherent with criteria $2$ of our desiderata, we expect to have no redundancy in the generator.
\begin{defi}
A generator $\Gs=\{E_1,...,E_l\}$ of $X$ is irreducible if $\Gs'\subset \Gs$ implies $\bigvee\Gs'\subset X$.
\end{defi}
Note that every generator can be turned into an irreducible one by recursively extracting ERCs that do not change the total closure, until no such extraction is possible. Since such extraction can be done in many ways in principle, we note that irreducible generators are not unique in general, neither have necessarily the same size. In our example above, note that $\{E_3\}$ and $\{E_1,E_2\}$ are irreducible generators of $E_3$.\\

 From now on, we will only consider ERC-generators unless explicitly stated, so we will refer to them simply as generators. 

\begin{lem}
Let $X\in\Csr_{con}$, and suppose $\Gs$ contains two ERCs $E_1,E_2$ such that $E_2\prec E_1$. Then $\Gs$ is not irreducible.
\label{lem:containment_out_ERC}
\end{lem}
\begin{proof}
Since $E_2\subset E_1$ then $E_1=\clos(E_1\cup E_2)=E_1\js E_2$. Therefore,
$$\bigvee\limits_{E\in \Gs} E = \bigvee\limits_{E\in \Gs, E\neq E_2} E=X.$$
\end{proof}

Lemma~\ref{lem:containment_out_ERC} show that ERCs organized as a hierarchy can be surprisingly useful to understand the inner generative structure of semi-organizations, and hence enable their efficient generation. In what follows, we develop in detail how synergies and complementarities are built within the ERC hierarchy, and hence represents the two forms of productive novelty underlying the generation of relevant semi-organizations (see section~\ref{EfficientGeneration}).

\section{Synergy}
\label{Synergy}
\subsection{Basic Concept}

The concept of synergy represents the generation of a reaction that a combination of closed sets triggers but no part can trigger alone.

\begin{defi}
Let $X,Y\in\Csr$ be incomparable.  We say the pair $\{X,Y\}$ is \emph{synergetic} if
$\Rs_{X\js Y}\supsetneq\Rs_X\cup\Rs_Y$; that is, their combination activates at least one
reaction that neither part activates on its own.  Any such reaction is called a
\emph{synergetic reaction of} $\{X,Y\}$.
\label{def:synergy_sets}
\end{defi}
At the level of ERCs, synergy takes the following form.
\begin{defi}
Let $\{E_1,\dots,E_k\}\subseteq\Es$ with $k\geq 2$, and write $X=\bigvee_{j=1}^{k}E_j$.
We say that $\{E_1,\dots,E_k\}$ is \emph{synergetic} if there exists a reaction
$r\in\Rs_{X}$ such that
$$r\notin\Rs_{\bigvee_{E\in\mathbf{S}}E}\qquad
  \text{for every proper subcollection }\mathbf{S}\subset\{E_1,\dots,E_k\}.$$
Equivalently, $\supp(r)\subseteq X$ but
$\supp(r)\not\subseteq\bigvee_{E\in\mathbf{S}}E$ for all
$\mathbf{S}\subset\{E_1,\dots,E_k\}$, so that the activation of $r$ requires
the joint contribution of all $k$ ERCs.  We call $r$ a \emph{synergetic reaction
generated by} $\{E_1,\dots,E_k\}$ and $\ERCs(r)$ the \emph{generated synergy}.
\label{synergy_general}
\end{defi}

\begin{defi}
The \emph{degree} of a synergetic collection $\{E_1,\dots,E_d\}$ (Definition~\ref{synergy_general}) is its cardinality $d$. Synergies of degree $2$ and $3$ are called \emph{binary} and \emph{ternary}, respectively.
\label{def:degree}
\end{defi}

An important property of synergies is that every synergy between closed sets traces back to a synergy among ERCs that are exclusive to each part.

\begin{lem}
\label{lem:synergy_closed_to_ERC}
Let $X,Y\in\Csr$ be incomparable and suppose $\{X,Y\}$ is synergetic, witnessed by $r\in\Rs_{X\js Y}-(\Rs_X\cup\Rs_Y)$.  Then, there exist $\mathbf{S}_X\subseteq\Es$ with $E\subseteq X$ for all $E\in \mathbf{S}_X$, and $\mathbf{S}_Y\subseteq\Es$ with $E\subseteq Y$ for all $E\in \mathbf{S}_Y$, and the pair of ERC collections $(\mathbf{S}_X,\mathbf{S}_Y)$ is synergetic also witnessed by $r$.
\end{lem}
\begin{proof}
For simplicity let $|\supp(r)|= 2$, write $\supp(r)=\{s_1,s_2\}$.  Because $r\notin\Rs_X$, $\supp(r)\not\subseteq X$; because $r\notin\Rs_Y$, $\supp(r)\not\subseteq Y$.  Hence the two species cannot both belong to $X$, nor both to $Y$.

We take $r$ to be the first reaction of $\Rs_{X\js Y}-(\Rs_X\cup\Rs_Y)$ to fire while computing $\clos(X\cup Y)$, so that $\supp(r)\subseteq X\cup Y$; hence we can assume that $s_1\in X- Y$ and $s_2\in Y- X$.  Since $X$ is reactive and $s_1\in X$, some reaction $r_1\in\Rs_X$ satisfies $s_1\in\supp(r_1)\cup\prods(r_1)$; set $E=\ERCs(r_1)=\clos(\supp(r_1))\subseteq X$, so $s_1\in E$.  If $E\subseteq Y$ then $s_1\in Y$, contradicting $s_1\in X- Y$; hence $E\not\subseteq Y$.  Symmetrically, pick $r_2\in\Rs_Y$ with $s_2\in\supp(r_2)\cup\prods(r_2)$ and set $E'=\ERCs(r_2)\subseteq Y$; then $E'\not\subseteq X$.

Now $\supp(r)=\{s_1,s_2\}\subseteq E\cup E'\subseteq E\js E'$, so $r\in\Rs_{E\js E'}$.  Since $s_2\notin X\supseteq E$, we have $r\notin\Rs_E$; since
$s_1\notin Y\supseteq E'$, we have $r\notin\Rs_{E'}$.  Therefore $\{E,E'\}$ is a synergetic pair for $r$.
For $|\supp(r)|=m>2$, apply the same construction to each $s_i\in\supp(r)$: since each $s_i$ lies in some part, pick $r_i$ in that part's reactions with $s_i\in\supp(r_i)\cup\prods(r_i)$ and set $E_i=\ERCs(r_i)$. Then $\supp(r)\subseteq\bigcup_i E_i\subseteq\bigvee_i E_i$, so $r\in\Rs_{\bigvee_i E_i}$. Take a minimal subcollection of $\{E_1,\dots,E_m\}$ whose join contains $\supp(r)$. By minimality no proper subcollection of it triggers $r$, so it is a synergetic collection for $r$; it contains ERCs from both $X$ and $Y$ because $\supp(r)\not\subseteq X$ and $\supp(r)\not\subseteq Y$, and it is partitioned between $X$ and $Y$ according to where each of its ERCs lies.

\end{proof}

Every synergy between closed sets therefore decomposes into a synergy between ERCs from each closed set.  This confirms that ERCs are the correct granularity at
which synergy arises. A binary synergy, in which $\{E_1,E_2\}$ generates $E_3$, arises when a reaction $r$ with $\supp(r)\subseteq E_1\cup E_2$ is not triggered by either $E_1$ or $E_2$ alone---the two ERCs must act jointly to expose the reaction.  A ternary synergy requires that none of the three pairwise joins triggers $r$, so all three reactant species come from three mutually distinct ERCs.  A degree-$d$ synergy (Definition~\ref{def:degree}) requires a reaction $r$ with $|\supp(r)|\geq d$, implying $\rho=\max_{r\in\Rs}|\supp(r)|$ represents a bound on the number of ERCs that shall be combined to produce a synergy in $(\Ms,\Rs)$.

Synergies were first introduced in~\cite{veloz2019existence} to lay the groundwork for understanding when combinations of closed sets create novel structure.

\subsection{Maximal and Fundamental Synergies}
Definition~\ref{synergy_general} does not clarify what are the inner generative causes and consequences of a synergy.  For example, consider the reaction network in table~\ref{tab:corrected_erc}, where we identify the reactions, ERCs that

\begin{table}[h]
\centering
\begin{minipage}[t]{0.52\linewidth}
\centering
\begin{tabular}{|c|c|c|}
\hline
Reaction & $E_i$ & $R_{E_i}$ \\ \hline
$r_1: a+e \to b+c$ & $E_1=\{a,b,c,e\}$ & $\{r_1,r_2\}$ \\
$r_2: b+c \to 2c$ & $E_2=\{b,c\}$ & $\{r_2\}$ \\
$r_3: d+a \to a+b$ & $E_3=\{a,b,d\}$ & $\{r_3\}$ \\
$r_4: g \to f+d$ & $E_4=\{d,f,g\}$ & $\{r_4,r_5\}$ \\
$r_5: f+d \to 2d$ & $E_5=\{d,f\}$ & $\{r_5\}$ \\
$r_6: c+d \to g$ & $E_6=\{c,d,f,g\}$ & $\{r_4,r_5,r_6\}$ \\ \hline
\end{tabular}
\end{minipage}
\hfill
\begin{minipage}[c]{0.47\linewidth}
\centering

\includegraphics[width=\linewidth]{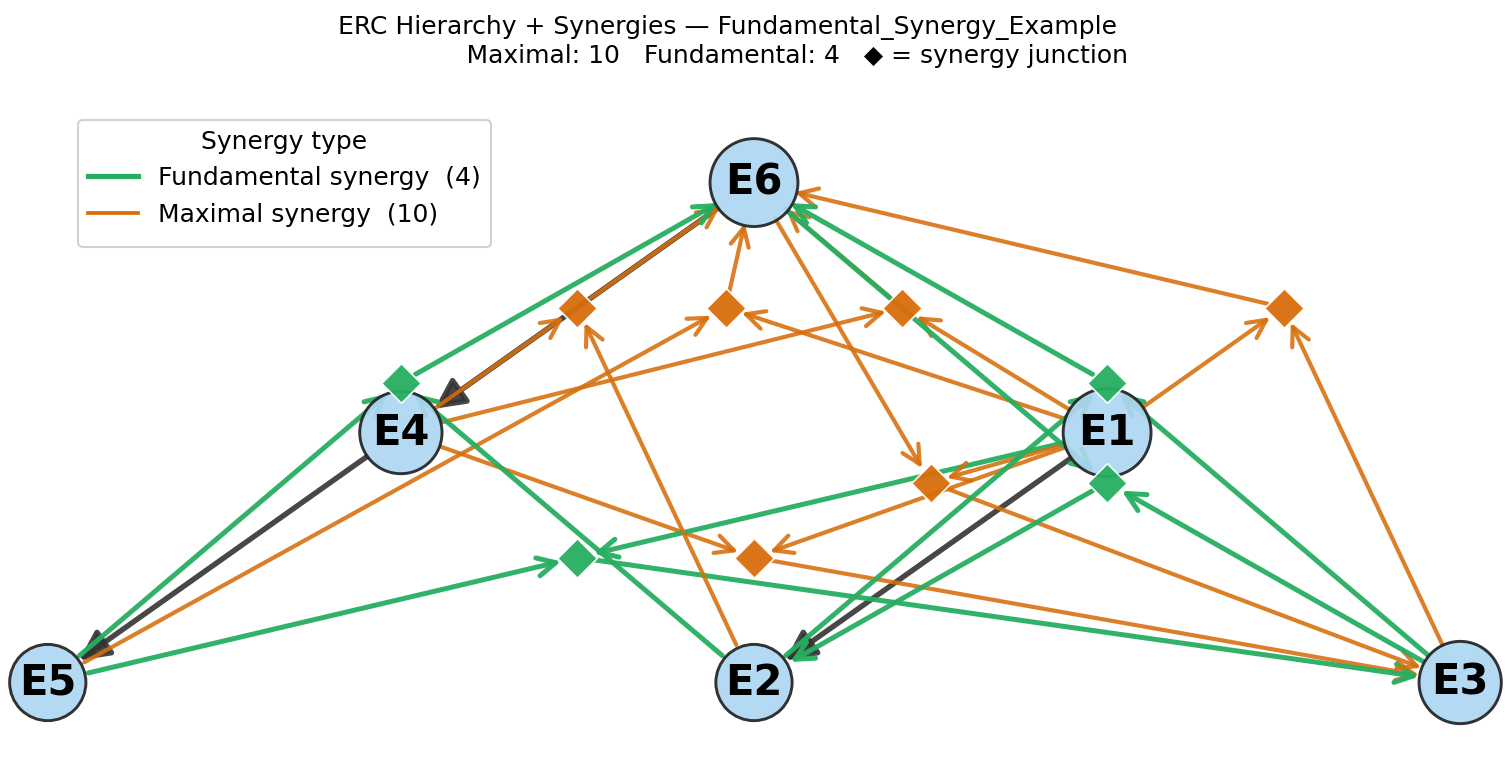}
\end{minipage}
\caption{Example of a reaction network with fundamental and non-fundamental synergies. Left: Reaction network table with ERCs and Reaction sets. Right: Image of the ERC hierarchy with synergetic relations.}
\label{tab:corrected_erc}
\end{table}

Note that $E_2\prec E_1$ and $E_5\prec E_4\prec E_6$. Consider the following synergies 
\begin{equation*}
\begin{split}
\sigma_1&:\{E_2,E_3\}\text{ generates } E_6,\quad
\sigma_2:\{E_2,E_3\}\text{ generates } E_5,\quad
\sigma_3:\{E_2,E_3\}\text{ generates } E_4,\\
\sigma_4&:\{E_1,E_3\}\text{ generates } E_6,\quad
\sigma_5:\{E_1,E_3\}\text{ generates } E_5,\quad
\sigma_6:\{E_1,E_3\}\text{ generates } E_4.
\end{split}
\end{equation*}

Since $E_5\prec E_4\prec E_6$ we have that $\sigma_1$ and $\sigma_4$ generate a larger ERC ($E_6$) than $\sigma_i$ and $\sigma_j$ with $i=2,3$ and $j=5,6$ respectively. This defines a maximality criteria of a generated synergy. 
\begin{defi}
A synergy of $\{E,E'\}$ generating $E_{syn}$ is called a maximal synergy if and only if there is no other synergy of $\{E,E'\}$ generating $E_{syn}'$ such that $E_{syn}\prec E_{syn}'$.
\label{def:maximal_synergy}
\end{defi}
We can generalize the definition to higher order synergies.
\begin{defi}
A degree-$d$ synergy in which $\{E_1,\dots,E_k\}$ generates $E_{syn}$ is \emph{maximal} if no other synergy of the same collection produces a strictly larger ERC:
$$\nexists\,E_{syn}'\text{ such that }\{E_1,\dots,E_k\}\text{ generates } E_{syn}'\text{ and }E_{syn}\subset E_{syn}'.$$
\label{def:maximal_synergy_k}
\end{defi}
Maximal synergies can be generated by combinations of ERCs at different containment levels. In our example, note that since $E_2\prec E_1$, we have that $\sigma_1$ needs a lower-level ERC ($E_2$) to combine with $E_3$ to generate $E_6$ than $\sigma_4$ (which requires $E_1$). This defines the minimality criteria for producers for a synergy. 
 \begin{defi}
A maximal degree-$d$ synergy in which $\{E_1,\dots,E_k\}$ generates $E_{syn}$ is \emph{fundamental} if for every index $i$ there is no $E_i'\subset E_i$ with $E_i'\in\Es$ such that $\{E_1,\dots,E_{i-1},E_i',E_{i+1},\dots,E_k\}$ is a maximal synergy generating $E_{syn}$.
We denote a fundamental degree-$d$ synergy by $\{E_1,\dots,E_k\}\syns E_{syn}$.
\label{fundamental_ERC_synergy}
\end{defi}

The following result shows that all synergies reduce to fundamental ones through the hierarchy.

\begin{lem}[Fundamental descent of a maximal synergy]
\label{lem:synergy_reduction}
If $\{E_1,\dots,E_k\}$ is a degree-$d$ maximal synergy generating $E_{syn}$,
then there exist $\hat E_i\subseteq E_i$, $i=1,\dots,k$, such that
$\{\hat E_1,\dots,\hat E_k\}\syns E_{syn}$ is a fundamental synergy.
\end{lem}
\begin{proof}
If the synergy is not fundamental at slot $i$, Definition~\ref{fundamental_ERC_synergy}
provides $E_i'\subset E_i$, $E_i'\in\Es$, such that
$\{E_1,\dots,E_{i-1},E_i',E_{i+1},\dots,E_k\}$ is still a maximal synergy
generating $E_{syn}$. Iterating this replacement produces, in each slot, a strictly
$\subset$-descending sequence in the finite poset $(\Es,\subseteq)$, which must
terminate. The terminal tuple $\{\hat E_1,\dots,\hat E_k\}$ admits no such
replacement in any slot and is therefore fundamental. Each original producer is
recovered along a saturated $\prec$-chain
$\hat E_i\prec\cdots\prec F_0^{(i)}=E_i$ in $(\Es,\prec)$.
\end{proof}

%

\section{Complementarity}
\label{Complementarity}

Complementarity is the second form of productive novelty, and the counterpart on the output side: a combination need not fire any new reaction, yet a species that one ERC requires may be produced by another, so the requirement for semi-self-maintenance of combination is fewer species than the requirement for self-maintenance demanded apart. 
\subsection{The Requirement of a Combination}
\label{ReqCombination}

Complementarity rests on one crucial fact: combining reactive closed sets cannot
increase the required species for semi-self-maintenance.

\begin{lem}
Let $X,Y\in\Csr_{con}$.  Then $\reqs(X\js Y)\subseteq\reqs(X)\cup\reqs(Y)$.
\label{req_containment}
\end{lem}
\begin{proof}
Let $w\in\reqs(X\js Y)$, so $w\notin\prods(\Rs_{X\js Y})$.  Since
$X\js Y=\clos(X\cup Y)$ and the closure adds only produced species,
$w\notin\prods(\Rs_{X\js Y})$ forces $w\in X\cup Y$; say $w\in X$.  As $X$ is
reactive, $w\in\supp(\Rs_X)\cup\prods(\Rs_X)$, and $w\notin\prods(\Rs_X)$ because
$\prods(\Rs_X)\subseteq\prods(\Rs_{X\js Y})$.  Hence
$w\in\supp(\Rs_X)-\prods(\Rs_X)=\reqs(X)$.
\end{proof}

We say that $X$ and $Y$ are \emph{complementary} (at the closed-set level) if the
combination strictly reduces the set of required species, i.e. $\reqs(X\js Y)\subset\reqs(X)\cup\reqs(Y)$.  By
Lemma~\ref{req_containment} this is equivalent to requiring that one part produces a
species the other requires:
$$\prods(\Rs_X)\cap\reqs(Y)\neq\emptyset
  \quad\text{or}\quad
  \prods(\Rs_Y)\cap\reqs(X)\neq\emptyset.$$

We now specify the vocabulary to define complementarity among ERCs.

\begin{defi}
For $E,E'\in\Es$, the \emph{supply} from $E$ to $E'$ is
$\supl{E}{E'}=\prods(\Rs_E)\cap\reqs(E')$.
\label{def:supply}
\end{defi}

\begin{defi}
$E,E'\in\Es$ are \emph{complementary} if
$\supl{E}{E'}\cup\supl{E'}{E}\neq\emptyset$.  A species $s$ is complementary in
$E\js E'$ if $s\in\supl{E}{E'}\cup\supl{E'}{E}$.
\label{complementary}
\end{defi}

Every complementarity between closed sets traces back to complementary ERCs, where the
producing ERC is necessarily exclusive to the producing part.

\begin{lem}
\label{lem:complementarity_closed_to_ERC}
Let $X,Y\in\Csr_{con}$ and suppose $s\in\prods(\Rs_X)\cap\reqs(Y)$.
Then there exist $E\in\Es$ with $E\subseteq X$, $E\not\subseteq Y$, and $E'\in\Es$
with $E'\subseteq Y$, such that $s\in\supl{E}{E'}$.
\end{lem}
\begin{proof}
Since $s\in\prods(\Rs_X)$, there exists $r_X\in\Rs_X$ with $s\in\prods(r_X)$; set
$E=\ERCs(r_X)=\clos(\supp(r_X))\subseteq X$, so $s\in\prods(\Rs_E)$.  If $E\subseteq Y$
then $\prods(\Rs_E)\subseteq\prods(\Rs_Y)$, giving $s\in\prods(\Rs_Y)$, contradicting
$s\in\reqs(Y)$; hence $E\not\subseteq Y$.

Since $s\in\reqs(Y)$, there exists $r_Y\in\Rs_Y$ with $s\in\supp(r_Y)$; set
$E'=\ERCs(r_Y)\subseteq Y$, so $s\in E'$.  Because $\Rs_{E'}\subseteq\Rs_Y$ and
$s\notin\prods(\Rs_Y)$, we have $s\notin\prods(\Rs_{E'})$.  Combined with
$s\in\supp(\Rs_{E'})$, this gives $s\in\reqs(E')$.  Therefore
$s\in\prods(\Rs_E)\cap\reqs(E')=\supl{E}{E'}$.
\end{proof}

Similar to synergy, complementarity at the closed-set level is therefore always witnessed by ERCs and thus complementarities can safely be encoded in the ERC hierarchy. However, and different to synergy, note two ERCs belonging to the same hierarchical chain can be complementary. We call such case \emph{intra-chain} complementarity: since $E'\subseteq E$ gives $\prods(\Rs_{E'})\subseteq\prods(\Rs_E)$, we have $\supl{E'}{E}=\prods(\Rs_{E'})\cap\reqs(E)\subseteq\prods(\Rs_E)\cap\reqs(E) =\emptyset$, so supply can run only from the larger ERC to the smaller. The other possible case is \emph{inter-chain} complementarity: producer and consumer lie on distinct branches of the ERC hierarchy, and supply crosses between them.

\subsection{Pure, Minimal and Fundamental Complementarities}
\label{FundComplementarity}

Lemma~\ref{req_containment} shows that a requirement can only be removed. Thus, a combination produces a required species in exactly one of two ways: one part already produces what another requires, or a synergy happening along the closure computation enables its production. However, two ERCs can be either synergetic, complementary, both, or none. Thus, there could be in principle various reasons happening in parallel to remove a requirement.  Consider the reaction network in table~\ref{tab:corrected_erc}. Note that $E_5$ and $E_3$ are complementary as $\supl{E_5}{E_3}=\{d\}$ and they have no synergy, while $E_3$ and $E_4$ are both synergetic (e.g. $\{E_3,E_5\}$ generates $E_6$) and complementary ($\supl{E_5}{E_3}=\{d\}$).  

Therefore, complementarity can be found as a direct reduction of requirements in a combination, or via synergies whose subsequent addition of products reduces the combination's requirements. For example, note that  $\reqs(E_4)=\{g\}$, and consider $(E_1,E_4)$, a non-complementary pair whose synergy produces $g$. Complementarities obtained through synergies can only be known once synergies are computed, so we focus on the complementarities we can know directly from the ERC requirements and productions. 
\begin{defi}
Two ERCs $E,E'$ are purely complementary if they are complementary and not synergetic.
\end{defi}

Similar to synergy, complementarity can also be found in its fundamental form. To illustrate this, note in our running example that $\supl{E_i}{E_5}=\{f\}$ and hence $E_i$ is complementary to $E_5$ for $i=4,6$. However, the reason why $f$ is supplied to $E_5$ is that $r_4$ produces $f$ and since $E_4\prec E_6$, then $E_6$ inherits the ability to supply $f$ to $E_5$ from $E_4$.  We formalize this in what follows. 

\begin{lem}
Let $E_0\subseteq E_1$ and $E'\in\Es$. Then $\supl{E_0}{E'}\subseteq\supl{E_1}{E'}$.
\label{supply_monotone}
\end{lem}
\begin{proof}
$E_0\subseteq E_1$ gives $\Rs_{E_0}\subseteq\Rs_{E_1}$, hence
$\prods(\Rs_{E_0})\subseteq\prods(\Rs_{E_1})$; intersecting with $\reqs(E')$ gives the
inclusion.
\end{proof}
 
A larger ERC activates more reactions, so its production only grows. Hence, the ERCs producing a fixed species are closed upward, and along each branch of the hierarchy there is a lowest ERC at which the species first appears as a product, and a lowest at which it is first required.
 
\begin{defi}
For $s\in\Ms$, the \emph{minimal producers} and \emph{minimal consumers} of $s$ are
$$\minprod(s)=\min{}_{\subseteq}\{E\in\Es:s\in\prods(\Rs_E)\},\qquad  \mincons(s)=\min{}_{\subseteq}\{E\in\Es:s\in\reqs(E)\}.$$
\label{def:minprodcons}
\end{defi}
 
As with the maximal synergies, $\minprod(s)$ and $\mincons(s)$ are antichains of ERCs drawn from the closures of the reactions touching $s$.
These minimal representatives are the canonical carriers of a complementary species, and minimising \emph{both} sides yields the fundamental
relation --- the complementary analogue of the fundamental synergy (Definition~\ref{fundamental_ERC_synergy}), which traces a maximal synergy to its
minimal reactant ERCs.

\begin{defi}
\label{def:fundamental_complementarity}
Let $s\in\supl{E}{E'}$ be a complementary species, so $E$ produces $s$ and $E'$ requires it. The complementarity is \emph{fundamental} if both sides are minimal,
$E\in\minprod(s)$ and $E'\in\mincons(s)$; we write $E\fcomp{s}E'$. (Minimising only the producing side, $E\in\minprod(s)$, gives a \emph{minimal productive
complement} of $E'$ on $s$ --- the producer-side half of the relation.)
\end{defi}
 
A fundamental complement carries one or many complementary species from the smallest ERC that produces it to the smallest that requires it in the ERC hierarchy. 
 
\begin{lem}[Fundamental descent of a complementarity]
\label{lem:complementarity_reduction}
Let $E,E'$ be a complementary pair. Then for each complementary species $s\in\supl{E}{E'}$ there is a fundamental complementarity $\hat E_s\fcomp{s}\hat E_s'$ with $\hat E_s\in\minprod(s)$, $\hat E_s\subseteq E$, $\hat E_s'\in\mincons(s)$ and $\hat E_s'\subseteq E'$ (symmetrically for $s\in\supl{E'}{E}$).
\end{lem}
\begin{proof}
Let $s\in\supl{E}{E'}$, so $s\in\prods(\Rs_E)$ and $s\in\reqs(E')$. Some
$r\in\Rs_E$ produces $s$ with $\ERCs(r)\subseteq E$; contracting along the
$\prec$-chain below $E$ to a $\subseteq$-minimal producer yields
$\hat E_s\in\minprod(s)$ with $\hat E_s\subseteq E$. Some $r'\in\Rs_{E'}$ has
$s\in\supp(r')$ with $\ERCs(r')\subseteq E'$; since $s\notin\prods(\Rs_{E'})$,
$\ERCs(r')$ does not produce $s$, so contracting to a $\subseteq$-minimal
consumer yields $\hat E_s'\in\mincons(s)$ with $\hat E_s'\subseteq E'$. Both sides
being minimal, $\hat E_s\fcomp{s}\hat E_s'$ is fundamental.
\end{proof}


\section{Productive Novelty: The Generative Structure of Persistence}
\label{ProductiveNovelty}

Lemmas~\ref{lem:complementarity_reduction} and~\ref{lem:synergy_reduction} enable characterizing synergies and complementarities among closed sets as fundamental synergies and fundamental complementarities among ERCs.

\begin{lem}[Fundamental Witness]
\label{lem:fund_witness}
Let $X, Y \in \Csr_{con}$ be incomparable.
\begin{enumerate}
\item If $\{X,Y\}$ is synergetic, there exist collections
  $\mathbf{S}_1 \subseteq \Es$ with $E \subseteq X$ for all $E \in \mathbf{S}_1$, and
  $\mathbf{S}_2 \subseteq \Es$ with $E \subseteq Y$ for all $E \in \mathbf{S}_2$,
  such that $\{\mathbf{S}_1, \mathbf{S}_2\} \syns E_{syn}$ for some $E_{syn} \in \Es$.
\item If $\{X,Y\}$ is complementary, there exist $E \in \Es$ with $E \subseteq X$,
  $E \not\subseteq Y$, and $E' \in \Es$ with $E' \subseteq Y$ such that
  $E \fcomp{s} E'$ for some $s$.
\end{enumerate}
\end{lem}
\begin{proof}
(1) Apply Lemma~\ref{lem:synergy_closed_to_ERC} to obtain ERC collections
from each part witnessing the synergy; Lemma~\ref{lem:synergy_reduction}
then descends each producer slot to a fundamental synergy, preserving the
partition into $\mathbf{S}_1 \subseteq X$ and $\mathbf{S}_2 \subseteq Y$.
(2) Apply Lemma~\ref{lem:complementarity_closed_to_ERC} to obtain $E \subseteq X$,
$E' \subseteq Y$ with $s \in \supl{E}{E'}$; Lemma~\ref{lem:complementarity_reduction}
descends to the fundamental complementarity.
\end{proof}

We now make use of this result to revisit the notion of generator, and apply it to the construction of semi-organizations.

\subsection{Fundamental Generators}
\label{FundamentalGenerators}

\begin{defi}[Fundamental extension]
\label{def:fund_extension}
Let $X \in \Csr_{con}$ and $\mathbf{G}$ be a generator of $X$. Let $\mathbf{S} \subseteq \Es$ be non-empty with
$E \not\subseteq X$ for every $E \in \mathbf{S}$.  The extension of $\mathbf{G'}=\mathbf{G}\cup\mathbf{S}$, which yields $X'=\bigvee \mathbf{G'}$, is
\emph{fundamental} if one of the following holds:
\begin{enumerate}
\item \emph{(Synergetic $|\mathbf{S}|\geq 1$)}\;  $\{X,\bigvee\mathbf{S}\}$ is a synergy and a fundamental synergy  $\{\mathbf{S}_X, \mathbf{S}\} \syns E_{syn}$ with $\mathbf{E}\in \mathbf{S}_X$ implies $E\subseteq X$; or
\item \emph{(Complementary, $|\mathbf{S}|=1, \mathbf{S}=\{E\}$)}\; $\{X, E\}$ are complementary with $E'\subset X$ such that $E'\fcomp{s} E$ be a fundamental complementarity for some $s\in E\cup E'$.
\end{enumerate}
\end{defi}

Note that a fundamental extension made of a single ERC could fulfill both fundamental step criteria.

\begin{defi}[Fundamental generator]
\label{def:fund_generator}
A generator $\Gs$ of $X$ is \emph{fundamental} if it is irreducible and its
elements can be arranged into an ordered partition $\Gs = \mathbf{S}_1 \cup \cdots \cup \mathbf{S}_m$
into non-empty, pairwise-disjoint blocks such that every extension of $X_{i-1}$
by $\mathbf{S}_i$ is fundamental, where $X_0 = E_\emptyset$ and
$X_i = X_{i-1} \js \bigl(\bigvee \mathbf{S}_i\bigr)$.
\end{defi}

Fundamental generators are a very restrictive class of the possible generators that can be in principle generated. We will now show they are sufficient to generate each and every productively novel semi-organization.

\subsection{Elementary Semi-Organizations and Generative Index}
\label{Elementary_Persistence_Modules}

Among all semi-organizations, those that contain no smaller semi-organization are special because no intermediate irreducible generator can possibly produce another semi-organization. 

\begin{defi}
\label{def:elementary}
$X \in \Sos$ is an \emph{elementary semi-organization} if the only
$Y \subset X$ with $Y \in \Sos$ is $Y = E_\emptyset$.
We assign it generative index $0$ and write $\Sos_0$ for the set of elementary
semi-organizations.
\end{defi}

Elementary semi-organizations are the simplest semi-organizations that can be generated, and non-elementary semi-organizations rest on top of smaller semi-organizations. We capture this stratification through a \emph{semi-organization generative index}.

\begin{defi}[Generative index]
\label{def:gen_index}
The \emph{generative index} $\mathrm{ord}(X)$ of $X\in\Sos$ is $0$ if $X\in\Sos_0$,
and otherwise
$$\mathrm{ord}(X)=1+\max\{\mathrm{ord}(Y): Y\in\Sos,\ E_\emptyset\subset Y\subset X\}.$$
We write $\Sos_i=\{X\in\Sos:\mathrm{ord}(X)=i\}$ for the semi-organizations of
generative index $i$, so that $\Sos=\bigcup_{i\geq0}\Sos_i$.
\end{defi}

It is worth distinguishing the generative index from the \emph{generative size} of a construction. The generative index $\mathrm{ord}(X)$ of Definition~\ref{def:gen_index} is fixed by the containment of semi-organizations inside $X$: it records how deeply semi-organizations are nested, and is a property of $X$ alone, independent of how $X$ is generated. The generative size, in contrast, is the number of ERCs --- equivalently, of combination steps --- that a particular generator requires to build $X$, and it does depend on the chosen generative pathways followed in the ERC hierarchy.

For example, consider the reaction network in Eq.~\eqref{RN-Ex1}, whose ERCs are $E_1=\{s_1,s_2\}$, $E_2=\{s_3\}$ and $E_3=\{s_1,s_2,s_3,s_4\}$. Here $E_1,E_2\in\Sos_0$ are elementary semi-organizations and $E_3\in\Sos_1$ properly contains both, so the generative index of $E_3$ is $1$ regardless of how it is built. Yet $E_3$ is itself a P-ERC, so it can be generated either directly by the singleton generator $\{E_3\}$ --- a generative size of one, using no combination step and exposing neither $E_1$ nor $E_2$ --- or by the synergetic combination $\{E_1,E_2\}$ --- a generative size of two, reaching $E_3$ in a single productive-novelty step from two elementary semi-organizations.

\subsection{Completing Generators via Fundamental Extensions}

Now we are finally in position to state the main result of this paper. Namely, we will show that the generator of a closed set lying between two generative indexes can safely be extended using fundamental extensions only to reach a semi-organization in the next generative index.

\begin{lem}[Fundamental extension step]
\label{lem:fund_ext_step}
Let $X\in\Sos$, $\Gs$ an irreducible generator of $X$, $\emptyset\neq \mathbf{C}_p\subset \Gs$,
$X_p=\bigvee_{E\in \mathbf{C}_p}E$, and $\Gs'=\Gs-\mathbf{C}_p$. If $\reqs(X_p)\neq\emptyset$, there is a
non-empty $\mathbf{S}\subseteq \Gs'$ whose extension of $X_p$ is fundamental. It is a
fundamental \emph{complementarity} when some $s\in\reqs(X_p)$ has
$s\in\prods(\Rs_{\bigvee \Gs'})$, and a fundamental \emph{synergy} otherwise.
\end{lem}
\begin{proof}
Pick $s\in\reqs(X_p)$; since $X$ is SSM, $s\in\prods(\Rs_X)$.

\emph{Case (a): $s\in\prods(\Rs_{\bigvee \Gs'})$.} Then $\bigvee \Gs'$ already produces
$s$, so some $E\in \Gs'$ (or an ERC below it, which we take instead) has
$s\in\prods(\Rs_E)$. As $s\in\reqs(X_p)$, $s\in\prods(\Rs_E)\cap\reqs(X_p)
=\supl{E}{X_p}$, so $(X_p,E)$ is complementary. Lemma~\ref{lem:fund_witness}(2)
descends this to a fundamental complementarity; the extension of $X_p$ by $\{E\}$
is fundamental.

\emph{Case (b): $s\notin\prods(\Rs_{\bigvee \Gs'})$.} Producing $s$ inside
$X=\clos(X_p\cup\bigvee \Gs')$ therefore requires a reaction outside
$\Rs_{X_p}\cup\Rs_{\bigvee \Gs'}$. Let $r$ be the first such reaction to fire while
closing $X_p\cup\bigvee \Gs'$. Before any reaction outside the two separate closures
fires, only species of $X_p\cup\bigvee \Gs'$ are present, so
$\supp(r)\subseteq X_p\cup\bigvee \Gs'$; and since $r\notin\Rs_{X_p}$ and
$r\notin\Rs_{\bigvee \Gs'}$, $\supp(r)$ meets both $X_p\setminus\bigvee \Gs'$ and
$\bigvee \Gs'\setminus X_p$. Take $\mathbf{S}\subseteq \Gs'$ minimal with
$\supp(r)-X_p\subseteq\bigvee \mathbf{S}$; then $\supp(r)\subseteq X_p\js(\bigvee \mathbf{S})$ and, by
minimality, $r\notin\Rs_E$ for every $E\in \mathbf{S}$. Hence
$r\in\Rs_{X_p\js(\bigvee \mathbf{S})}-(\Rs_{X_p}\cup\bigcup_{E\in \mathbf{S}}\Rs_E)$: the extension is
synergetic, and Lemma~\ref{lem:fund_witness}(1) descends it to a fundamental
synergy.
\end{proof}

Lemma~\ref{lem:fund_ext_step} shows if we take apart some ERCs from the generator of a semi-organization, then it is certain that some of them produce a fundamental extension to the remaining generator. That is an intermediate step for the following result.

\begin{thm}[Fundamental completion of semi-organizations]
\label{thm:fund_completion}
Let $X_i\in\Sos_i$ and $X_{i+1}\in\Sos_{i+1}$ with $X_i\subset X_{i+1}$. Then, for any irreducible generator
$\Gs$ of $X_{i+1}$ and any $\mathbf H\subseteq\Gs$ with
$X_i\subset\bigvee\mathbf H\subset X_{i+1}$ and $\bigvee\mathbf H\in\Csr_{con}$,
one has $\reqs(\bigvee\mathbf H)\neq\emptyset$, and $X_{i+1}$ is obtained from
$\bigvee\mathbf H$ by fundamental extensions only.
\end{thm}
\begin{proof}
If $Y\in\Sos$ satisfies $X_i\subset Y\subset X_{i+1}$, then $X_i\subset Y$
gives $\mathrm{ord}(Y)\geq\mathrm{ord}(X_i)+1=i+1$, while $Y\subset X_{i+1}$ gives
$\mathrm{ord}(Y)\leq\mathrm{ord}(X_{i+1})-1=i$ (Definition~\ref{def:gen_index}), a
contradiction; hence no such $Y$. Now $\bigvee\mathbf H$ is connected and closed with
$X_i\subset\bigvee\mathbf H\subset X_{i+1}$, so it is not a semi-organization
(it would lie strictly between $X_i$ and $X_{i+1}$),  so
$\reqs(\bigvee\mathbf H)\neq\emptyset$. Lemma~\ref{lem:fund_ext_step} applies with
$\mathbf C_p=\mathbf H$ and yields a fundamental extension toward $X_{i+1}$;
finiteness of $\Gs$ gives termination at $X_{i+1}$.
\end{proof} 

Theorem~\ref{thm:fund_completion} is a powerful consequence of the theoretical construction of the ERC hierarchy and the fundamental synergies and complementarities. It demonstrates that generator extensions using fundamental extensions only, are sufficient to arrive to the next semi-organization in $\Sos$, if any such higher semi-organization exists. 

Moreover, the following results imply that fundamental extensions can also generate $\Sos_0$. 
 
 \begin{lem}[No novelty implies P-ERC]
\label{lem:no_novelty_PERC}
Let $X \in \Sos_0$.  If no sub-collection of $\{E \in \Es : E \subseteq X\}$
is synergetic and no two such ERCs are complementary, then $X$ is a P-ERC.
\end{lem}
\begin{proof}
Without synergy or complementarity,
$\reqs(X) = \bigcup_{E \subseteq X,\,E \in \Es}\reqs(E)$.
Since $X \in \Sos_0$ is SSM, $\reqs(X) = \emptyset$, so $\reqs(E) = \emptyset$
for every $E \subseteq X$ in $\Es$: each is a P-ERC.
Being of generative index $0$ then forces $X$ itself to be the unique such P-ERC.
\end{proof}

\begin{cor}
$\Sos_0$ can be computed using P-ERCs and fundamental generators only.
\label{lem:Sos_0}
\end{cor}
\begin{proof}
From the P-ERCs, we use the containment relation to select those belonging to $\Sos_0$. Now let $X\in\Sos_0$ be not a P-ERC, and let $\Gs$ be an irreducible generator of $X$. Since an ERC that is SSM is a P-ERC, $X$ is not an ERC, so for any $E\in\Gs$ we have $E_\emptyset\subset E\subset X$ and we can start from $\mathbf{H}=\{E\}$. For any $\mathbf{H}\subset\Gs$ with $E_\emptyset\subset\bigvee\mathbf{H}\subset X$ and $\bigvee\mathbf{H}\in\Csr_{con}$, Definition~\ref{def:elementary} implies that $\bigvee\mathbf{H}\notin\Sos$, hence $\reqs(\bigvee\mathbf{H})\neq\emptyset$. Lemma~\ref{lem:fund_ext_step} with $\mathbf{C}_p=\mathbf{H}$ then yields a non-empty $\mathbf{S}\subseteq\Gs-\mathbf{H}$ whose extension of $\bigvee\mathbf{H}$ is fundamental. Replacing $\mathbf{H}$ by $\mathbf{H}\cup\mathbf{S}$ and iterating, finiteness of $\Gs$ ensures that $X=\bigvee\Gs$ is reached after finitely many fundamental extensions.
\end{proof}

From here, we have the necessary pieces to build an algorithm to compute organizations. Starting from the ERC hierarchy equipped with the fundamental synergies and complementarities, we can safely build $\Sos_0$, and next by extending such semi-organizations, so they are not SSM and hence we can apply theorem~\ref{thm:fund_completion}, we can complete such generators to reach the next semi-organization index using fundamental extensions. However, one challenge remains, which is extending semi-organizations themselves to explore new levels. We analyze this challenge in the next section. 

\subsection{The Challenge of Extending Semi-organizations: Two Conjectures}
\label{Extending_SO}
Extending the generator of a semi-organization implies adding one or many ERCs to it. Note they must be connected to our semi-organization, which imposes a first restriction but still leaves too much room open for combinatorial explosion. 
A first idea for reducing the target ERCs for expanding the generators of a semi-organization is to repeat the idea that worked previously, i.e. extending the generator of a semi-organization using fundamental extensions only. However, this idea fails in some cases. For example, consider the following reaction network:
\begin{equation}
r_1=a\to b,~r_2=b\to a,~r_3=x+c\to x+d+a,~r_4=d+z\to c+z.
\end{equation}
We have 
\begin{itemize}
\item $E_0=\textrm{ERC}(r_1)=\textrm{ERC}(r_2)=\{a,b\}$ (P-ERC)
\item $E_1=\textrm{ERC}(r_3)=\{a,b,c,d,x\}$ ($\reqs(E_1)=\{c\}$)
\item $E_2=\textrm{ERC}(r_4)=\{c,d,z\}$ ($\reqs(E_2)=\{d\}$).
\end{itemize}

Therefore, $\Sos_0=\{E_0\}$, and $\Sos_1=\{E_1\js E_2\}$ with $E_1\js E_2=\{a,b,c,d,x,z\}$. Note $E_0$ has no fundamental extension available, but $E_0\prec E_1$ and $E_1\fcomp{\{c,d\}}E_2$. Therefore, $E_0$ is a semi-organization whose generator cannot reach to $E_1\js E_2$ by means of fundamental extensions. 

This example shows that, even though fundamental extensions are sufficient to reach semi-organizations from non-semi-organizations, they are not sufficient to extend semi-organizations. We conjecture that by allowing the containment relationship as a special kind of extension in these cases might resolve the problem. Note direct containments represent the edges in the ERC hierarchy. In this example, we would `lift' our generator from $E_0$ to $E_1$, and with extending the possibilities for fundamental extensions, while keeping an irreducible generator. Although this idea might resolve the algorithmic exploration of the ERC hierarchy to generate semi-organizations, its implementation requires to reconsider fundamental relations as chain rather than ERC defined relations, so higher ERCs in a hierarchy can inherit the fundamental relations in their contained ERCs. That might turn computationally expensive.

Alternatively, we can consider the construction path from $E_0$ closed, and restart a different exploratory path from $E_1$ using fundamental extensions only. The latter enables discovery $E_1\js E_2$ for a different reason. We simply explore generators from each ERC individually, without considering containment relationships. We conjecture that such exploration might as well generate all semi-organizations. However, it is unclear how to keep track of the semi-organization index through this construction process, as each semi-organization is obtained from a different generative path. Therefore, computing the generative indexes of organizations could become an expensive computational task.  

These two conjectures to compute semi-organizations imply different theoretical constructions for demonstrating the completeness and efficiency of an algorithm. Thus, we leave such endeavour for future work, and focus on  showing that the combinatorial structure of fundamental synergies and complementarities, i.e. the combinatorial structure required in Theorem~\ref{thm:fund_completion}, is very light compared with the combinatorial structure of the ERCs, as evidence for the efficiency of future algorithms based on our results. 

\section{Synergy and Complementarity in Biological Reaction Networks}
\label{RelevantExtensions}

Theorem~\ref{thm:fund_completion} shows that fundamental synergies and complementarities are sufficient to extend non-semi-organizations to reach semi-organizations. As a first illustration of how small is the number of fundamental extensions versus the whole possible ERC combinations, we show the ERC hierarchy of a reaction network with $31$ reactions in Figure~\ref{fig:synergy_complementarity} (left synergies, right complementarities).

\begin{figure}[h]
\centering
\includegraphics[width=0.90\linewidth]%
{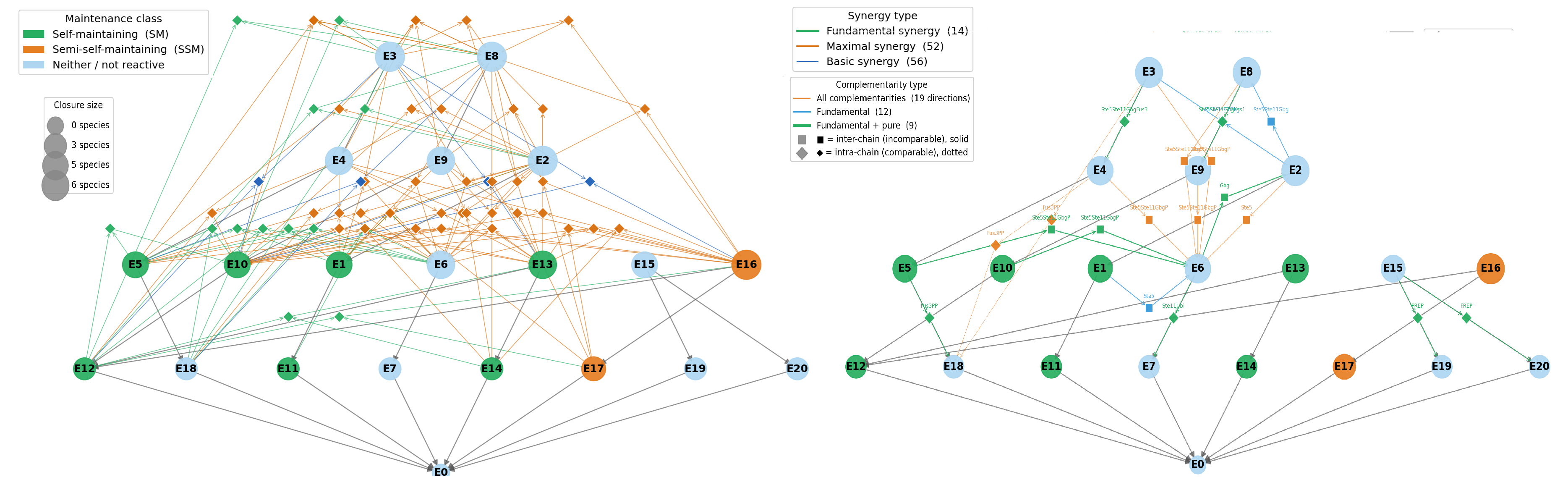}
\caption{{\bf ERC Hierarchy enhanced with Synergies and Complementarities:} ERC hierarchy for Biomodel~237
  {\it Schaber2006\_Pheromone\_Starvation\_Crosstalk},
  annotated with synergy junctions (diamond nodes, left) and complementarity
  junctions (square nodes, right).
  Colour encodes synergy class: green = fundamental, orange = maximal/minimal,
  blue = simple.}
\label{fig:synergy_complementarity}
\end{figure}

Simple inspection shows that fundamental synergies (14) and pure complementarities (9) are much sparser than their basic counterparts (56 and 19 respectively). In what follows we perform a thorough analysis across biological databases to quantify how these fundamental structures scale with the reaction network size.

\subsection{Dataset and analytical approach}
\label{sec:growth_dataset}

We analysed \(438\) reaction networks drawn from two sources: \(426\) models from
the BioModels repository~\cite{BioModels2020} (covering metabolic, signalling,
gene-regulatory, cell-cycle, circadian, apoptotic, and immune sub-classes) and
\(12\) genome-scale metabolic reconstructions from the BiGG
database~\cite{schellenberger2010bigg} (the largest in our analysis).  All networks satisfy \(|\Es|\geq 4\) after excluding \(E_\emptyset\); reactions
range from \(4\) to \(1{,}000\) (median~17) and ERC count from \(4\) to
\(642\) (median~9). Algorithms to build the ERCs hierarchy from the reaction network, as well as for verifying synergies and complementarities of different kinds, are all brute force.

Growth tendencies are estimated by fitting power laws
\(y = A \cdot x^{\alpha}\) via OLS on log-log data.
For raw counts \(\alpha\) describes absolute growth; for normalised fractions of
\(\binom{|\Es|}{k}\) (\(k=2\) binary synergies and complementarities, \(k=3\) ternary synergies) the fitted exponent
\(\beta\) captures density relative to all possible \(k\)-subsets and is
estimated by a separate regression on the fraction data directly.
Networks in which a property is absent are excluded from the respective fit.

\subsection{ERC hierarchy}
\label{sec:erc_growth}

\begin{figure}[h]
  \centering
  \includegraphics[width=\linewidth]{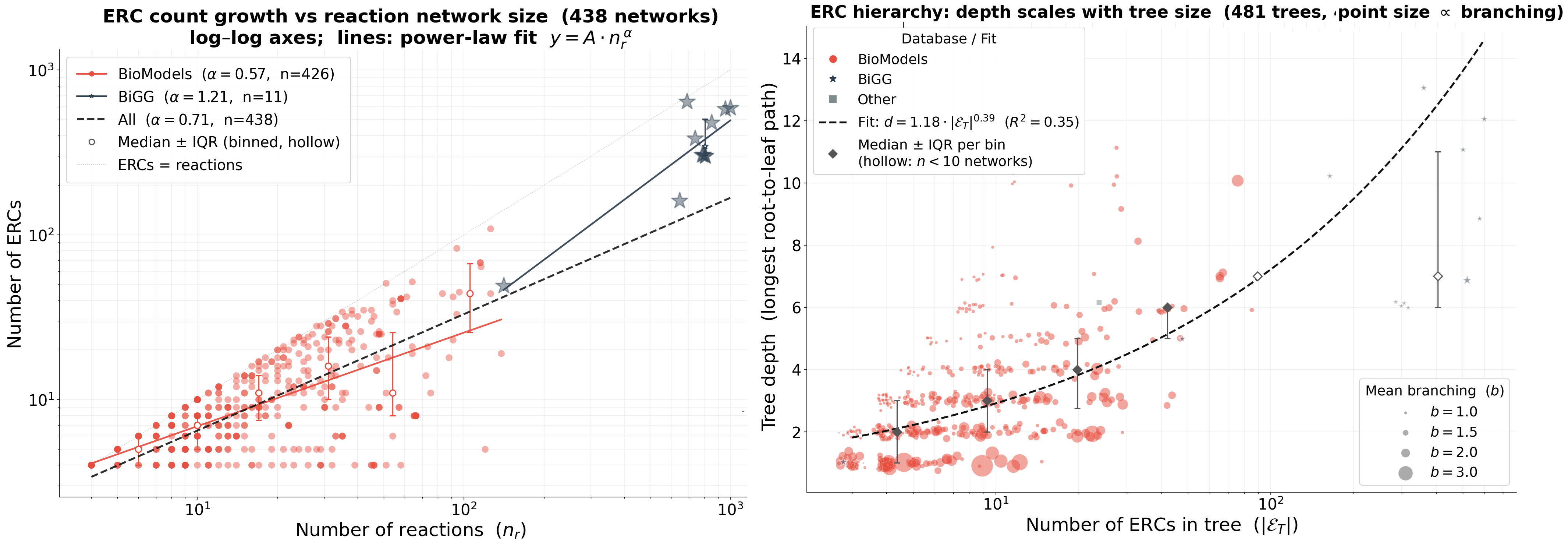}
  \caption{ERC hierarchy structure across \(438\) biological networks
    (BioModels and BiGG, \(|\Es|\geq 4\)).
    \textbf{Left:} ERC count \(|\Es|\) as a function of the number of
    reactions \(n_r\) (log--log axes); the dashed diagonal marks
    \(|\Es|=n_r\).
    \textbf{Right:} depth of each connected component of the ERC hierarchy (longest chain of direct containments, excluding $E_\emptyset$; labelled trees in the figure) versus mean
    branching average number of directly contained ERCs per non-minimal ERC});
    point size proportional to \(|\Es|\).
  \label{fig:erc_growth}
\end{figure}

Figure~\ref{fig:erc_growth} characterises how the ERC hierarchy scales with network size.
The left panel shows that ERCs are consistently fewer than reactions, and often
substantially so: the count grows sublinearly
(\(\alpha_{\mathrm{all}}\approx 0.71\), \(R^{2}=0.60\)), so doubling the
number of reactions increases \(|\Es|\) by only a factor of roughly \(1.6\).
BioModels networks display an even shallower slope
(\(\alpha_{\mathrm{BioMD}}\approx 0.57\)), while the twelve BiGG
reconstructions suggest a super-linear trend
(\(\alpha_{\mathrm{BiGG}}\approx 1.21\)), due to small sampling.
The aggregated sublinearity reflects a structural integration of larger networks: the ERC
representation becomes increasingly compressed relative to the raw reaction
network's size.

The right panel plots, for
each qualifying connected component of the ERC hierarchy after removing $E_\emptyset$ (labelled trees in the figure), its number of ERC nodes ($\Es_{\cal T}$) against its depth;
marker size encodes mean branching so that near-chain components ($b \approx 1$)
appear as small dots and highly branching components as large ones.
Depth grows with component size following a sublinear power law,
\begin{equation}\label{eq:depth_scaling}
  d \;\approx\; 1.18 \cdot \Es_{\cal T}^{\,0.39}
  \qquad (R^{2}=0.35,\; p \ll 10^{-10}),
\end{equation}
with exponent $\gamma \approx 0.39$ (95\,\% CI: $0.35$--$0.44$).
A balanced tree of constant branching $b$ would yield $d \sim \log_b \Es_{\cal T}$
(logarithmic growth, $\gamma \to 0$); a pure chain would yield $d = \Es_{\cal T} - 1$
($\gamma = 1$).
The observed intermediate exponent places ERC hierarchy components between these extremes:
they are compressed relative to chains but more elongated than a balanced
tree of the same mean branching would predict.
Crucially, marker sizes remain uniformly small across the full horizontal
range: branching does not intensify as components grow larger, consistent with
the near-zero raw correlation $r(\Es_{\cal T}, b) \approx 0$.
Together, these two results describe a \emph{compressively scaling} ERC
hierarchy: more ERCs produce deeper hierarchies at a sub-square-root pace, while
the internal branching structure stays locked near one to two children per
internal node regardless of component size.

\subsection{Synergy and complementarity}
\label{sec:syn_comp_growth}

Figure~\ref{fig:growth_tendency} reports the growth of synergy and
complementarity normalised fractions
\(f = \mathrm{count}/\binom{|\Es|}{k}\) as functions of
\(|\Es|\).

\begin{figure}[h]
  \centering
  \includegraphics[width=\linewidth]{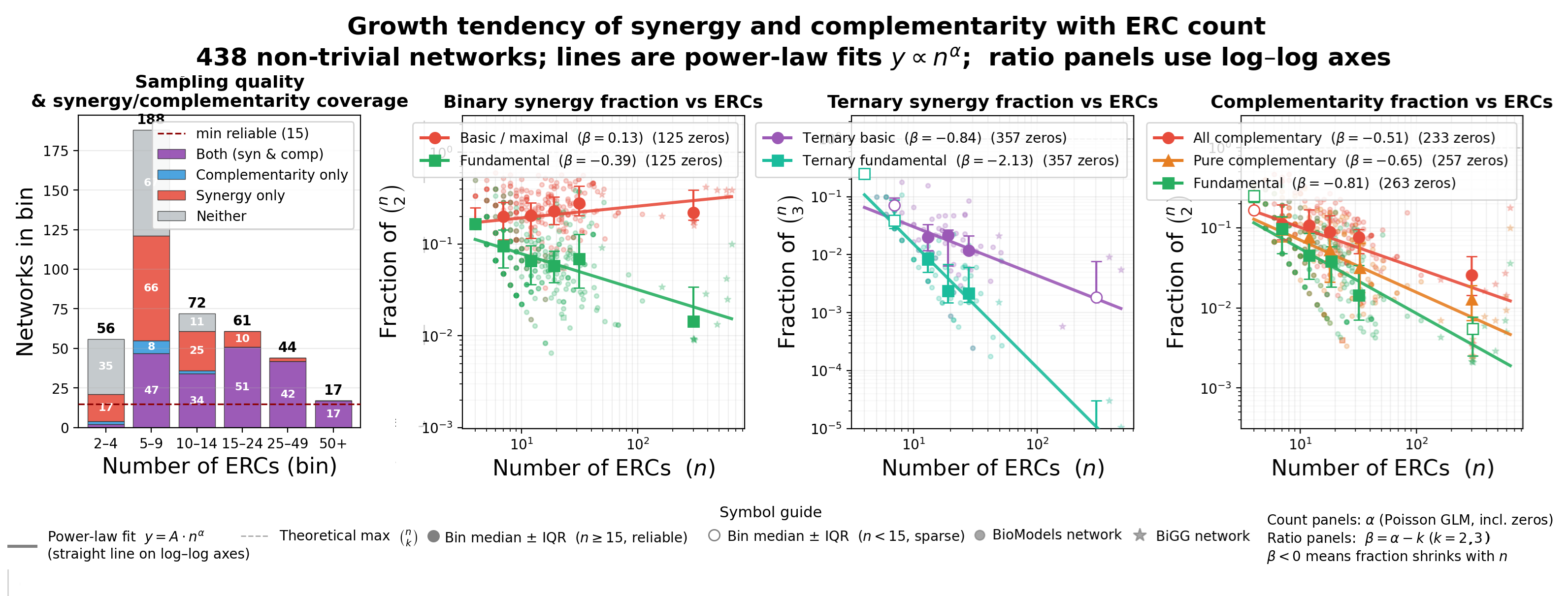}
  \caption{Growth of synergy and complementarity 
    normalised fractions as functions of ERC count
    (\(|\Es| \geq 4\), \(438\) networks).
    Solid coloured lines are power-law fits; hollow markers are binned medians
    \(\pm\) IQR.
    Networks with count \(= 0\) are excluded from the fits.}
  \label{fig:growth_tendency}
\end{figure}

First, it is important to mention that synergy and complementarity are pervasive in complex reaction networks. The top-left plot shows that every single reaction network with more than 15 ERCs encodes synergy and all except one above 25 ERCs encode complementarities. Since the number of ERCs is a lower bound on the number of reactions, these networks have more than 15 and 25 reactions, respectively. The normalised fraction plots reveal the tendencies. Two contrasting patterns emerge.  Basic synergy constitutes a roughly constant fraction of all ERC pairs regardless of network size (\(\beta\approx+0.13\), nearly flat): the probability that a randomly chosen pair is synergetic is
approximately scale-invariant.  Fundamental synergies display a clear and significant negative trend
(\(\beta\approx -0.4\)): as networks grow, the density of essential synergetic information shrinks.  Ternary synergies (ERC triples jointly enabling a target not reachable by any pairwise sub-collection) are rarer still --- present in only 18\% of networks versus 71\% for binary synergy --- and their fundamental fraction declines steeply, consistent with the expectation that higher-order
irreducible generators become increasingly sparse as network size grows.

Complementarity fractions decline even more pronouncedly.  All three complementarity measures --- basic, pure, and fundamental --- show negative \(\beta\), with fundamental complementarity exhibiting the steepest fall.

Taken together, these results establish a consistent pattern: the ERC count grows sublinearly with reaction network size, and the proportion of fundamental relationships among those ERCs shrinks further still. Larger biological reaction networks achieve persistence through a smaller and more selective set of irreducible generative relationships, while the vast majority of ERC combinations become structurally redundant.

\section{Conclusion}
This work builds a novel framework to understand the inner generative structure of chemical organizations: We introduced multiple notions:  generator (and its irreducible case), ERC hierarchy, synergetic and complementarity relations, their fundamental versions, and others. This framework enables proving a number of formal results to characterize irreducible generators made out of fundamental extensions as a way to obtain semi-organizations efficiently, being theorem~\ref{thm:fund_completion} our main result. Thus, we present here both theoretical insight into reaction network's inner generative structures and a foundation for efficient algorithms to compute semi-organizations. 

The theoretical construction is inspired by the notion of {\it productive novelty}, which captures situations in which the persistence of a combined module cannot be inferred directly from the persistence of its constituent parts. We further show that productive novelty arises through two fundamental mechanisms—synergy and complementarity (formalized in Sections~\ref{Synergy} and~\ref{Complementarity}). These concepts were refined to provide a principled perspective on how increasingly complex persistent organizations can emerge from irreducible sequences of ERCs. Reading a semi-organization as a minimal skeleton of such relationships therefore offers a principled decomposition of a reaction network into its irreducible functional units, of the kind sought in the study of autocatalytic cores, metabolic modularity, and the origins of self-maintaining organization.

We confirmed that the fundamental structures underlying semi-organizations are radically smaller than traditional complexity measures based on reaction count (see figure~\ref{fig:erc_growth}). More strikingly, Figure~\ref{fig:growth_tendency} shows that the proportion of \emph{fundamental} synergies and complementarities among all possible ERC pairs shrinks as networks grow, even as raw counts increase: basic synergy constitutes a roughly constant fraction of ERC pairs, while fundamental
synergy and all forms of complementarity become progressively sparser relative to the combinatorial ceiling $\binom{|\Es|}{k}$.

Our results suggest that efficient algorithms to build semi-organizations can be built, but we do not prescribe algorithms yet. Such endeavour shall address questions such as how to identify fundamental synergies or complementarities efficiently, and how to extend the generators of semi-organizations (see section~\ref{Extending_SO}). Turning our theoretical results into a practical algorithm for computing all semi-organizations of large reaction networks is left as future work.

Besides the challenges and opportunities in the quest for an efficient algorithmic framework, we would like to mention some of the areas where our results could find immediate application. RAF theory faces critical limitations in identifying minimal autocatalytic cores; current methods cannot decompose these networks systematically~\cite{xavier2020autocatalytic}. Peng et al.~\cite{peng2022hierarchical} noted that databases are tiny compared to chemistry's immensity and that algorithms struggle with combinatorial explosion. The RAF--Chemical Organizations link~\cite{hordijk2018autocatalytic} suggests that our framework, by isolating the fundamental synergetic and complementary pairs within an ERC hierarchy, could provide principled decompositions of autocatalytic sets into their minimal generative cores. Decomposing target networks into fundamental generators could reveal minimal synergistic sets, and the empirical sparsity of fundamental synergies means such sets, despite the increasing complexity of reaction networks in areas such as origins of life, metabolic analysis, and drug discovery, will be small in practice.

We believe that the formal methods developed here can significantly expand the capacities we have to analyze and understand the inner workings of life.

\section*{Data Availability Statement}

This work used BioModels and BiGG data publicly available online, our script to compute statistical analysis are accessible at \url{https://github.com/tveloz/pyCOT} (folder \texttt{projects/Generative\_Structure\_Orgs})

\section{Acknowledgments}

\bibliographystyle{plain}
\bibliography{biblio}

\end{document}